\documentclass[a4paper,twocolumn,10pt,unpublished]{quantumarticle}
\pdfoutput=1
\usepackage[utf8]{inputenc}
\usepackage[english]{babel}
\usepackage[T1]{fontenc}
\usepackage{amsmath,amssymb,amsfonts,amsthm}
\usepackage[colorlinks=false,hidelinks]{hyperref}
\usepackage{cleveref}

\usepackage{tikz}
\usepackage{lipsum}
\usepackage{braket}
\usepackage{todonotes}
\usepackage{quantikz}
\usepackage{caption}
\usepackage{subcaption}
\usepackage{makecell}
\usepackage{xcolor}

\usepackage[shortlabels]{enumitem}

\usepackage[numbers,sort&compress]{natbib}

\newtheorem{theorem}{Theorem}

\newtheorem{lemma}{Lemma}

\newcommand{\repeattheorem}[1]{%
  \begingroup
  \renewcommand{\thetheorem}{\ref{#1}}%
  \expandafter\expandafter\expandafter\theorem
  \csname reptheorem@#1\endcsname
  \endtheorem
  \endgroup
}

\NewEnviron{reptheorem}[1]{%
  \global\expandafter\xdef\csname reptheorem@#1\endcsname{%
    \unexpanded\expandafter{\BODY}%
  }%
  \expandafter\theorem\BODY\unskip\label{#1}\endtheorem
}

\DeclareMathOperator*{\tr}{Tr}

\DeclareMathOperator*{\locc}{LOCC}
\DeclareMathOperator*{\cptn}{CPTN}

\newcommand{\ketbra}[2]{|#1\rangle\langle#2|}
\newcommand{\ketbras}[1]{\ketbra{#1}{#1}}

\newcommand{\iu}{\mathrm{i}\mkern1mu}

\usepackage{tikz-layers}

\begin{document}

\title{Joint Wire Cutting with Non-Maximally Entangled States} 

\author{Marvin Bechtold}
\affiliation{Institute of Architecture of Application Systems, University of Stuttgart, Universitätsstraße 38, 70569 Stuttgart, Germany}
\orcid{0000-0002-7770-7296}
\email{bechtold@iaas.uni-stuttgart.de}
\author{Johanna Barzen}
\email{barzen@iaas.uni-stuttgart.de}
\affiliation{Institute of Architecture of Application Systems, University of Stuttgart, Universitätsstraße 38, 70569 Stuttgart, Germany}
\orcid{0000-0001-8397-7973}
\author{Frank Leymann}
\affiliation{Institute of Architecture of Application Systems, University of Stuttgart, Universitätsstraße 38, 70569 Stuttgart, Germany}
\orcid{0000-0002-9123-259X}
\email{leymann@iaas.uni-stuttgart.de}
\author{Alexander Mandl}
\affiliation{Institute of Architecture of Application Systems, University of Stuttgart, Universitätsstraße 38, 70569 Stuttgart, Germany}
\orcid{0000-0003-4502-6119}
\email{mandl@iaas.uni-stuttgart.de}
\author{Felix Truger}
\affiliation{Institute of Architecture of Application Systems, University of Stuttgart, Universitätsstraße 38, 70569 Stuttgart, Germany}
\orcid{0000-0001-6587-6431}
\email{truger@iaas.uni-stuttgart.de}

\maketitle

\begin{abstract}
Distributed quantum computing leverages the collective power of multiple quantum devices to perform computations exceeding the capabilities of individual quantum devices.
A currently studied technique to enable this distributed approach is wire cutting, which decomposes a quantum circuit into smaller subcircuits by cutting their connecting wires.
These subcircuits can then be executed on distributed devices, and their results are classically combined to reconstruct the original computation's result.
However, wire cutting requires additional circuit executions to preserve result accuracy, with their number growing exponentially with each cut.
Thus, minimizing this sampling overhead is crucial for reducing the total execution time.
Employing shared non-maximally entangled~(NME) states between distributed devices reduces this overhead for single wire cuts, moving closer to ideal teleportation with maximally entangled states.
Extending this approach to jointly cutting multiple wires using NME states remained unexplored.
Our paper addresses this gap by investigating the use of NME states for joint wire cuts, aiming to reduce the sampling overhead further. 
Our three main contributions include (i)~determining the minimal sampling overhead for this scenario, (ii)~analyzing the overhead when using composite NME states constructed from smaller NME states, and (iii)~introducing a wire cutting technique that achieves the optimal sampling overhead with pure NME states, paving the way towards wire cutting with arbitrary NME states.

\end{abstract}

\section{Introduction}
Despite the potential of quantum computing to efficiently tackle classically intractable problems~\cite{Shor1997,Liu2021a}, contemporary quantum devices encounter several limitations, such as a restricted number of available qubits, and scaling up to larger devices remains a significant challenge~\cite{Preskill2018,Leymann2020a}. 
In response, distributed quantum computing architectures offer a promising approach for scalability, leveraging the collective power of multiple quantum devices to perform computations that would be infeasible for a single device~\cite{Bravyi2022,Ang2022,Awschalom2021,Barral2024,Furutanpey2023a}. 
These quantum devices can be interconnected either through classical means~\cite{Avron2021,Dunjko2018} or, in the foreseeable future, by quantum links~\cite{Cuomo2020,Khait2023a,Niu2023a}.
The realization of quantum connectivity is expected to evolve from short-range quantum links between adjacent chips of modular quantum devices in the near term~\cite{Gold2021, Conner2021} to long-range quantum connections via cables between different devices as the technology advances~\cite{Zhong2019}.

Quantum teleportation, which enables the exchange of quantum states between remote quantum devices, is an essential building block for distributing quantum computations in the long term~\cite{Awschalom2021,Cuomo2020}.
However, its implementation necessitates maximally entangled states shared between distributed quantum devices.
These maximally entangled states enable the transfer of qubits between the devices by exchanging just two classical bits per qubit~\cite{Bennett1993}.
While generating entangled states among multiple commercial quantum devices is currently beyond reach, it is expected to become achievable in the foreseeable future~\cite{Bravyi2022,Choi2023}.
Furthermore, generating entangled states will be susceptible to noise~\cite{Gold2021,Zhong2019,Narla2016}, yielding \textit{non-maximally entangled~(NME)} states, i.e., states of intermediate level of entanglement, that complicate the use of teleportation, as error-free state transmission relies on perfect maximally entangled states~\cite{Prakash2012}.

In the interim, circuit cutting offers a practical alternative for distributing quantum computations without relying on maximally entangled states~\cite{Barral2024}.
It partitions a quantum circuit into multiple variations of subcircuits that can be executed on separate quantum devices, with the results combined to replicate the original circuit's outcome~\cite{Bechtold2023_cuttingPatterns}. 
Wire cutting is a specific form of circuit cutting.
It involves interrupting wires in a quantum circuit, each representing the state of a qubit, through repeated measurements and qubit initializations~\cite{Peng2019,Lowe2023,Harada2023}. 
This facilitates the distribution of the execution of a large circuit between multiple smaller devices by cutting the connecting wires.
However, wire cutting requires an increased number of shots to achieve the same accuracy in the expectation value as the original circuit execution~\cite{Brenner2023}.
These additional shots, known as \emph{sampling overhead}, increase multiplicatively per cut, resulting in exponential growth.
Therefore, minimizing the sampling overhead of each cut is crucial.

Recent findings show that allowing classical communication between subcircuits during wire cutting can significantly lower the sampling overhead~\cite{Brenner2023}. 
Additionally, cutting multiple wires jointly, rather than cutting each wire individually, further reduces the total overhead~\cite{Brenner2023}. 
Moreover, as first quantum links between devices are anticipated, using NME states in wire cutting has been shown to reduce the sampling overhead for cutting a single wire, offering a more accessible alternative to default quantum teleportation that requires maximally entangled states~\cite{Bechtold2024}. 

Building on these developments, our research focuses on jointly cutting multiple wires in a quantum circuit using NME states.
This approach aims to reduce the sampling overhead further and extend the application of NME states beyond cutting single wires.
We concentrate on parallel wire cuts, which involve cutting multiple wires within the same time slice of a quantum circuit.
However, we also argue that our results on the optimal sampling overhead can be applied to arbitrary wire cuts without limiting the cuts to the same time slice.
Our work presents three main contributions.
(i)~We determine the optimal sampling overhead, quantifying the required additional circuit executions for a parallel cut of multiple wires in a quantum circuit using NME states in \Cref{theorem_overhead}. 
This theorem demonstrates how NME states reduce the sampling overhead in this context, extending previous known sampling overheads that were limited to scenarios without utilizing NME states.
(ii)~We show how to cut $n$ wires given a set of available entangled states in \Cref{theorem_cut_composite_states,theorem_composite_max_ent,theorem_composite_separable}.
By combining individual entangled states for a single parallel cut, the sampling overhead is reduced compared to separate cuts using individual states, except for maximally entangled parts of the composite state.
As a step towards a wire cutting procedure that supports arbitrary NME states, (iii)~we present in \Cref{theorem_decomposition} a procedure using pure NME states to cut multiple wires in parallel that achieves the optimal sampling overhead.

This work is organized as follows:
\Cref{sec:preliminaries} introduces the preliminaries necessary for \Cref{sec:main} that presents the main contributions.
Following this, \Cref{sec:discussion} discusses the results, and \Cref{sec:related_work} presents related work.
Finally, \Cref{sec:conclusion} concludes the work.

\section{Preliminaries}\label{sec:preliminaries}
This section outlines the necessary preliminaries. 
First, \Cref{sec:nme_states} focuses on quantifying entanglement in NME states.
\Cref{sec:mubs} then introduces \emph{mutually unbiased bases}~(MUBs), as operators constructed from MUBs are later employed in the presented wire cutting technique.
Next, \Cref{sec:quasi_prob_sim,sec:quasi_prob_sim_states} cover the quasiprobability simulation of non-local operators and maximally entangled states, respectively.
These concepts underpin wire cutting, which is presented in \Cref{sec:wire_cutting}.
Lastly, \Cref{sec:teleportation} explores quantum teleportation for multiple qubits, including the errors encountered when using NME states.

\subsection{Quantifying Entanglement in NME States}\label{sec:nme_states}
Entanglement is a fundamental quantum mechanical phenomenon observed in bipartite quantum systems described by their combined Hilbert space $A \otimes B$ for subsystems $A$ and $B$.
The state of such a composite quantum system is represented by a density operator $\rho_{AB}$, which is a positive Hermitian operator with unit trace acting on $A \otimes B$.
The set of all density operators of a quantum system $A \otimes B$ is denoted as $D(A\otimes B)$.
The optional subscript in $\rho_{AB}$ clarifies the involved Hilbert spaces when necessary.
The density operator corresponding to a pure state $\ket{\psi} \in A \otimes B$ is represented by $\psi = \ketbras{\psi}\in D(A\otimes B)$. 

In bipartite systems, a state $\rho_{AB} \in D(A\otimes B)$ is called separable if it can be decomposed into independent states of its subsystems, meaning there exist $\rho_A \in D(A)$ and $\rho_B \in D(B)$ such that $\rho_{AB} = \rho_A \otimes \rho_B$.
The set of all separable states between $A$ and $B$ is referred to as $S(A,B) \subset D(A \otimes B)$.
Conversely, any state $\rho_{AB} \in D(A\otimes B)$ that does not allow for such decomposition, i.e., $\rho_{AB} \notin S(A,B)$, is called entangled.

Superoperators govern the evolution of density operators in quantum systems. 
These linear operators act on the density operators in $D(A\otimes B)$.
To be physically feasible, a superoperator $\mathcal{E}$ must satisfy two criteria: (i) it has to be completely positive and (ii) it must not increase the trace of any density operator it acts upon, i.e., $0 \le \tr[\mathcal{E}(\rho)] \le \tr[\rho]$ for any $\rho \in D(A\otimes B)$~\cite{Nielsen2009}.
The set of all \textit{completely positive trace-nonincreasing~(CPTN)} superoperators is represented by $\cptn(A\otimes B)$.

To characterize the degree of entanglement, \emph{local operations and classical communication~(LOCC)} play a crucial role~\cite{Chitambar2014}.
The corresponding set of operators $\locc(A,B) \subset \cptn(A\otimes B)$ encompasses local transformations on the subsystems $A$ and $B$, coordinated through classical communication.
An operator $\Lambda\in\locc(A, B)$ cannot create entanglement between $A$ and $B$.
Consequently, the degree of entanglement of the transformed state $\Lambda(\rho_{AB})$ remains at most equal to that of $\rho_{AB}$.
A function that satisfies this property can quantify entanglement and is called an entanglement monotone~\cite {Vidal2000}.
The \emph{generalized robustness of entanglement} is such an entanglement monotone~\cite{Steiner2003,Harrow2003}, which is used to quantify entanglement in the following.
It determines the minimal amount of an arbitrary state that must be mixed with $\rho_{AB}\in D(A\otimes B)$ to become a separable state~\cite{Steiner2003}, which is given by
\begin{align}\label{eq:robustness}
    \begin{split}
        R(\rho_{AB}) := \inf \biggl\{ \lambda \,\bigg|\,& \frac{\rho_{AB} + \lambda \omega_{AB} }{1+\lambda} \in S(A,B),\\
        &\omega_{AB} \in D(A\otimes B), \lambda \in \mathbb{R}_{\ge 0}  \biggr\}.
    \end{split}
\end{align}
Thus, $R(\rho_{A,B})=0$ if and only if $\rho_{A,B}\in S(A,B)$.
Moreover, for bipartite states $\rho_{A,B}$, the robustness value is upper bounded by $R(\rho_{A,B})\le \min(\dim(A),\dim(B))-1$~\cite{Steiner2003,Vidal1999}.

In the following, let $A$ and $B$ each comprise $n$ qubits, denoted as $A = A_0 \otimes ... \otimes A_{n-1}$ and $B = B_0 \otimes ... \otimes B_{n-1}$, where $A_i$ and $B_i$ represent the two-dimensional Hilbert spaces of individual qubits. 
We use $\ket{\Phi_n}_{AB}$ to represent the maximally entangled state between $A$ and $B$ with $R(\Phi_n) = 2^n-1$, defined as
\begin{align}\label{eq:max_ent_state}
\ket{\Phi_n}_{AB} &= \frac{1}{\sqrt{2^n}}\sum_{i=0}^{2^n-1} \ket{i}_{A}\otimes\ket{i}_{B} \\
&= \bigotimes_{i=0}^{n-1} \frac{1}{\sqrt{2}}(\ket{0}_{A_i}\otimes\ket{0}_{B_i} + \ket{1}_{A_i}\otimes\ket{1}_{B_i})\\ 
&= \bigotimes_{i=0}^{n-1} \ket{\Phi_1}_{A_iB_i}.\label{eq:n_bell_pairs}
\end{align}
The state $\ket{\Phi_n}_{AB}$ can be constructed using $n$ maximally entangled qubit pairs between $A$ and $B$, as shown in \Cref{eq:n_bell_pairs}.
To avoid subscript collisions in the density operator notation, we use either $\Phi_{AB}$, where its size $n$ is implicitly given by the systems $A$ and $B$, or $\Phi_n$ if the systems are clear from the context or irrelevant. 

In the following, we describe pure NME states used in our wire cutting procedure and give an explicit expression for their robustness of entanglement.
By making use of the Schmidt decomposition~\cite{Nielsen2009}, any entangled pure $2n$-qubit state $\ket{\psi}$ can be expressed as 
\begin{align}\label{eq:schmidt}
    \ket{\psi}_{AB} &= \sum_{i=0}^{2^n-1} \alpha_i \ket{\xi_i}_{A} \otimes \ket{\zeta_i}_{B} 
\end{align}
with $\alpha_i \in \mathbb{R}_{\geq 0}$, $\sum_{i} \alpha_i^2 = 1$, and orthonormal sets of $n$-qubit states $\{\ket{\zeta_i}\}$ and $\{\ket{\xi_i}\}$.
We refer to $\alpha := (\alpha_0, ..., \alpha_{2^n-1})$ as the \emph{Schmidt vector} consisting of the Schmidt coefficients.
It is worth noting that the Schmidt decomposition of any bipartite pure state can be computed on near-term quantum devices, given a sufficient number of instances of the state~\cite{BravoPrieto2020}.
For a pure state $\ket{\psi}_{AB}$, the robustness of entanglement is completely defined by its Schmidt vector $\alpha$~\cite{Steiner2003}:
\begin{align}\label{eq:robustness_pure_states}
    R(\psi_{AB}) &= \left(\sum_{i=0}^{2^n-1} \alpha_i\right)^2 - 1.
\end{align}

The representation in \Cref{eq:schmidt} shows that any $2n$-qubit state can be reformulated as 
\begin{align}\label{eq:Psi_alpha}
    \ket{\psi}_{AB} &= (U_A \otimes U_B) \ket{\Psi^{\alpha}}_{AB}.
\end{align}
Herein, $U_A= \sum_{i=0}^{2^n-1} \ketbra{\xi_i}{i}$ and $U_B = \sum_{i=0}^{2^n-1} \ketbra{\zeta_i}{i}$ are local unitary basis transformations from the computational basis into the basis of the Schmidt decomposition, and $\ket{\Psi^{\alpha}}$ is 
\begin{align}\label{eq:nme}
    \ket{\Psi^{\alpha}}_{AB} = \sum_{i=0}^{2^n-1} \alpha_i \ket{i}_{A} \otimes \ket{i}_{B}. 
\end{align}
Thus, the results obtained with the states $\ket{\Psi^{\alpha}}_{AB}$ can be extended to arbitrary pure states by applying the local unitary transformation $U_A \otimes U_B$ given in \Cref{eq:Psi_alpha}.
Therefore, in the following, we solely focus on NME qubit pairs $\ket{\Psi^{\alpha}}_{AB}$  of the form in \Cref{eq:nme}.

\subsection{Mutually Unbiased Bases}\label{sec:mubs}
Two orthonormal bases denoted as $\{\ket{u_j}\}_{j=0}^{d-1}$ and $\{\ket{v_j}\}_{j=0}^{d-1}$, of a $d$-dimensional Hilbert space are called mutually unbiased if all fidelities between their elements, which measure the closeness between these elements, have the same magnitude~\cite{Wootters1989}:
\begin{align}
\left|\braket{u_j|v_k}\right|^2 = \frac{1}{d} \quad \forall j,k \in \{0, ...,d-1\}.
\end{align}
This implies that if a system is prepared in a state from the first basis, then measuring any state in the second basis is equally likely.
This concept can be generalized to a set of $m$ different orthonormal bases $\{\mathcal{B}_j\}_{j=0}^{m-1}$  with $\mathcal{B}_j=\{\ket{e^j_l}\}_{l=0}^{d-1}$ where $\ket{e^j_l}$ is the $l$-th element of basis $\mathcal{B}_j$.
This set of bases is called mutually unbiased if all bases are pairwise mutually unbiased.

In this work, we use the specific construction from Durt et al.~\cite{Durt2010}, which is based on Galois fields introduced in \Cref{sec:galois_fields}, to generate a set of $2^n +1$ MUBs~$\{\mathcal{B}_j\}_{j=0}^{2^n}$ containing the computational basis as $\mathcal{B}_{2^n}$.
Each basis consists of $2^n$ elements, and the construction details are provided in \Cref{sec:phase_shift_operator}.
To facilitate the cutting of $n$ wires, we will employ $2^n$ unitary basis transformations $\{U_j\}_{j=0}^{2^n-1}$, each of which transforms the computational basis into a distinct MUB $\mathcal{B}_j$.
These transformations are defined as
\begin{align}\label{eq:mub_basis_transformation}
    U_j = \sum_{l=0}^{2^n-1}\ketbra{e_l^j}{l}.
\end{align}

An efficient algorithm exists for constructing these unitary basis transformations for $n$-qubit systems with $\mathcal{O}(n^3)$ time complexity~\cite{Wang2023a}. 
The corresponding circuits involve a maximum of $(n^2 + 7n)/2$ gates, including only Hadamard $H$, phase $S$, and $CZ$ gates.

\subsection{Quasiprobability Simulation of Non-Local Operators}\label{sec:quasi_prob_sim}
This section introduces the concept of quasiprobability simulations of non-local operators, a framework in which all circuit cutting techniques can be understood~\cite{Brenner2023}.
Within this framework, cutting a circuit distributed across subsystems $A$ and $B$ involves replicating its outcome by probabilistically substituting the non-local operators acting between these subsystems with LOCC~\cite{Piveteau2023}.

For this probabilistic substitution of the non-local operator \mbox{$\mathcal{E}\in \cptn(A\otimes B)$}, it is decomposed as
\begin{align}\label{eq:QPD}
    \mathcal{E} = \sum_{i=1}^{m} c_i \mathcal{F}_i
\end{align}
with $\mathcal{F}_i\in \locc(A, B)$~\cite{Piveteau2023}.
The coefficients $c_{i}\in \mathbb{R}$ are permitted to be negative but must satisfy the condition $\sum_{i=1}^{m} c_{i} = 1$. 
Due to possible negative coefficients, this is termed a quasiprobability distribution.
Consequently, such a decomposition of the operator $\mathcal{E}$ is known as \emph{quasiprobability decomposition~(QPD)}.
\looseness=-1

Using this QPD, the expectation value of the non-local operator $\mathcal{E}$, applied to a quantum state $\rho$ and measured with observable $O$, can be formulated as
\begin{align}\label{eq:expectation}
\tr[O\mathcal{E}(\rho)] = \kappa\sum_{i=1}^{m} p_{i} \tr[O\mathcal{F}_{i}(\rho)] \operatorname{sign}(c_i) 
\end{align}
using the probability distribution $p_i := |c_{i}|/\kappa$ with $\kappa := \sum_i |c_i|$ and denoting the sign of coefficient $c_i$ by $\operatorname{sign}(c_i)$~\cite{Brenner2023}.
\Cref{eq:expectation} enables estimating the expectation value $\tr[O\mathcal{E}(\rho)]$ via a Monte Carlo simulation  using only LOCC operators $\mathcal{F}_{i}\in \locc(A,B)$ as follows~\cite{Brenner2023}.
For each shot of the simulation, an index $i$ is randomly chosen with probability $p_{i}$, and the circuits corresponding to $\mathcal{F}_{i}$ are executed.
The measurement outcome, determined by the observable $O$, is then weighted by $\operatorname{sign}(c_i) \kappa$ and stored. 
Repeating this process and summing the weighted outcomes provides an estimate of $\tr[O\mathcal{E}(\rho)]$.

However, while this estimator preserves the expectation value, it increases the variance by a factor $\kappa$, known as the \emph{sampling overhead}.
To achieve a fixed statistical error $\epsilon$, the number of shots required scales as $\mathcal{O}(\kappa^2/\epsilon^2)$~\cite{Temme2017}.
Thus, finding a QPD with minimal sampling overhead $\kappa$ is crucial.
The optimal sampling overhead for decomposing $\mathcal{E}$ into operators $\mathcal{F}_i\in \locc(A,B)$ is denoted by $\gamma(\mathcal{E})$~\cite{Piveteau2023}.
It is defined as \looseness=-1
\begin{equation}\label{eq:min_sampling_overhead}
    \begin{split}
        \gamma(\mathcal{E}) := \min\biggl\{\sum_i |c_i| \,\bigg|\,&  \mathcal{E} = \sum_i c_i \mathcal{F}_i, c_i\in\mathbb{R},\\
        &\mathcal{F}_i \in \locc(A,B)\biggr\}.
    \end{split}
\end{equation}

When this method is applied to a circuit containing multiple non-local operators $\{\mathcal{E}_i\}_{i=0}^{n-1}$ that require decomposition, the resulting sampling overhead, achieved by utilizing the optimal QPD for each operator individually,  is given by the product of their individual sampling overheads~\cite{Brenner2023}, that is,
\begin{align}
    \prod_{i=0}^{n-1}\gamma(\mathcal{E}_i).
\end{align}

\subsection{Quasiprobability Simulation of Maximally Entangled States}\label{sec:quasi_prob_sim_states}
Quasiprobability simulation has been introduced to simulate non-local quantum operations by decomposing their superoperator into a linear combination of local operations and classical communication. 
Based on this concept, the measurement statistics of the maximally entangled state $\Phi_{AB}$ can be replicated by sampling only from NME states. 
The key idea is to find a superoperator $\mathcal{E}\in \cptn(A\otimes B)$ that converts an given NME state $\rho_{AB}$ into $\Phi_{AB}$, i.e., $\mathcal{E}(\rho_{AB})=\Phi_{AB}$~\cite{Piveteau2023}.
Since $\mathcal{E}$ increases the entanglement, it is not from the set $\locc(A,B)$~\cite{Chitambar2014}.

\begin{figure*}
    \centering
    \includegraphics{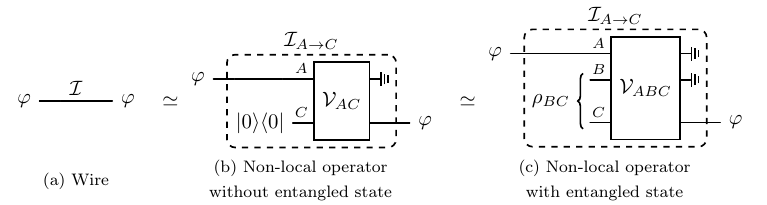}
    \caption{A wire (a) can be represented by the identity $\mathcal{I}_{A\rightarrow C}$ between systems $A$ and $C$, with a non-local operator $\mathcal{V}$. In (b), no entangled state is consumed by $\mathcal{V}$ and $\mathcal{I}_{A\rightarrow C}$ is obtained by tracing out $A$, symbolized by $\begin{quantikz}[column sep=4pt]\qw&\ground{}\end{quantikz}$~\cite{Brenner2023}. In (c), $\mathcal{V}$ consumes an NME state $\rho$ to produce $\mathcal{I}_{A\rightarrow C}$ when both $A$ and $B$ are traced out~\cite{Bechtold2024}. Adapted from~\cite{Bechtold2024}.}
    \label{fig:wire_cut_model}
\end{figure*}

Nevertheless, a QPD of $\mathcal{E}$ as given in \Cref{eq:QPD} can be used to decompose the transformation in a linear combination of operators from $\locc(A,B)$.
Based on this, a QPD of the maximally entangled state $\Phi_{AB}$ employing the NME state $\rho_{AB}$ can be formulated as
\begin{align}
    \Phi_{AB} &= \mathcal{E}(\rho_{AB})\\
    &= \sum_{i=1}^m c_i \mathcal{F}_i(\rho_{AB}) \label{eq:QPD_state}
\end{align}
where $\mathcal{F}_i \in \locc(A,B)$, and therefore the entanglement of the NME state $\mathcal{F}_i(\rho_{AB})$ does not exceed that of $\rho_{AB}$, i.e., $R(\mathcal{F}_i(\rho_{AB}))\le R(\rho_{AB})$.
By employing this QPD for $\Phi_{AB}$, the Monte Carlo simulation from \Cref{eq:expectation} can be used to estimate the expectation value for an observable $O$ by sampling from NME states $\mathcal{F}_i(\rho_{AB})$.

Based on this, the optimal sampling overhead for the quasiprobability simulation of the maximally entangled state $\Phi_{AB}$, with NME state $\rho_{AB}$ serving as a resource, is defined as
\begin{align}
    \hat{\gamma}^{\rho_{AB}}(\Phi_{AB}) := \min\left\{ \gamma(\mathcal{E}) \,\middle|\,   \mathcal{E}(\rho_{AB}) = \Phi_{AB} \right\}.
\end{align}
According to Takagi et al.~\cite[Proposition 7]{Takagi2024}, this is given by
\begin{align}\label{eq:optimal_overhead}
    \hat{\gamma}^{\rho_{AB}}(\Phi_{AB}) = \frac{2^{n+1}}{R(\rho_{AB}) + 1} - 1,
\end{align}
where $R(\rho_{AB})$ represents the generalized robustness of entanglement as defined in \Cref{eq:robustness}. 

This formulation also encompasses the particular scenario where separable states are employed to simulate the maximally entangled state $\Phi_{AB}$, which was studied in \cite{Brenner2023,Piveteau2023}.
In such instances, the optimal sampling overhead, denoted as $\hat{\gamma}(\Phi_{AB})$, is defined as
\begin{align}
    \hat{\gamma}(\Phi_{AB}) :&= \hat{\gamma}^{(\ketbras{00})^{\otimes n}}(\Phi_{AB}) \\
    &= 2^{n+1} - 1 \label{eq:sampling_overhead_sep_state}
\end{align}
where \Cref{eq:sampling_overhead_sep_state} results from the fact that $R((\ketbras{00})^{\otimes n})=0$.

\subsection{Wire Cutting}\label{sec:wire_cutting}

This section introduces wire cuts with and without NME states.
Moreover, for the scenario without NME states, an optimal QPD for parallel wire cutting is presented.

\subsubsection{Wire Cuts without NME states}
Wire cutting is the quasiprobability simulation of a non-local identity operator $\mathcal{I}_{A \rightarrow C}$, enabling the transfer of a state from the Hilbert space $A$ of a sending quantum device to the Hilbert space $C$ of another receiving device. 
The label $B$ for the Hilbert space is reserved to denote the sender's part of a shared entangled state later on.
The identity operator $\mathcal{I}_{A \rightarrow C}$ is realized by applying a non-local operator $\mathcal{V}$ and subsequently tracing out the sender's systems~\cite{Brenner2023,Bechtold2024}, as depicted in \Cref{fig:wire_cut_model}.
Thus, the minimal sampling overhead for $\mathcal{I}_{A \rightarrow C}$ corresponds to the minimal sampling overhead for quasiprobabilistically simulating such a $\mathcal{V}$.

For a wire cut without entanglement, the operator $\mathcal{V}$ operates on $A\otimes C$, as shown in \Cref{fig:wire_cut_model}~(b), and hence is denoted as $\mathcal{V}_{AC}$. 
This operator acts as the non-local identity $\mathcal{I}_{A \rightarrow C}$ between systems $A$ and $C$ once the sender's system $A$ is traced out~\cite{Brenner2023}:
\begin{align}\label{eq:id_without_NME}
\forall \varphi \in D(A): \operatorname{Tr}_{A}[\mathcal{V}_{AC}(\varphi \otimes \ketbra{0}{0}_C)] = \varphi,
\end{align}
where $\ketbra{0}{0}_C$ refers to the initial state of Hilbert space $C$.
As the state $\varphi$ transitions between different systems, its Hilbert space subscript is omitted. 
By employing a QPD for the non-local operator $\mathcal{V}_{AC}$, i.e., $\mathcal{V}_{AC}=\sum_i c_i \mathcal{F}_i$ with $\mathcal{F}_i \in \locc(A,C)$, wire cutting can be understood as the quasiprobabilistic simulation of the identity operator $\mathcal{I}_{A \rightarrow C}$: 
\begin{align}\label{eq:qpd_without_NME}
\forall \varphi \in D(A): \sum_i c_i \operatorname{Tr}_{A}[\mathcal{F}_i(\varphi \otimes \ketbra{0}{0}_C)] = \varphi.
\end{align}
Recent work established that the minimal sampling overhead for cutting a single wire, where $A$ and $C$ are two-dimensional Hilbert spaces, is $\gamma(\mathcal{I_{A \rightarrow  C}})=3$~\cite{Brenner2023}.
Henceforth, we will define the sampling overhead of the non-local identity operator $\mathcal{I}_{A \rightarrow C}$ simply as $\gamma(\mathcal{I})$, omitting the Hilbert space notation for brevity.

A parallel cut is a joint wire cut that involves the simultaneous cutting of $n$ wires in the same time slice of the circuit.
This process is described by the non-local quasiprobability simulation of the $n$-qubit identity operator $\mathcal{I}^{\otimes n}_{A \rightarrow C}$, where $A$ and $C$ are $2^n$-dimensional Hilbert spaces. 
The optimal sampling overhead for such a parallel cut is thus determined by $\gamma(\mathcal{I}^{\otimes n}) = 2^{n+1} - 1$~\cite{Brenner2023}.
As a result, a single joint wire cut of multiple wires has a lower sampling overhead compared to optimal individual cuts per wire when using classical communication without entanglement, since $\gamma(\mathcal{I}^{\otimes n}) = 2^{n+1} - 1 < 3^n = \left(\gamma(\mathcal{I})\right)^n$ for $n \ge 2$.

\begin{figure*}
    \centering
    \includegraphics{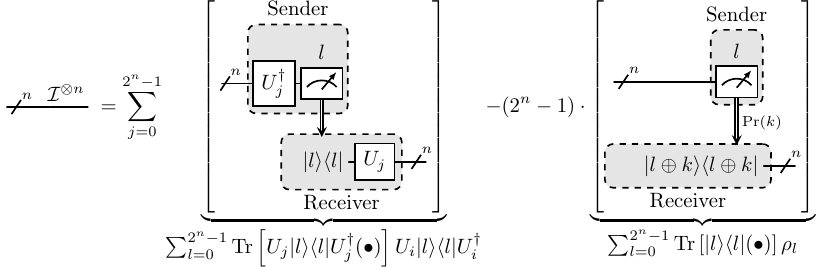}
    \caption{Optimal parallel cut of $n$ wires~\cite{Harada2023}.}
    \label{fig:wire_cut_harada}
\end{figure*}

\subsubsection{Wire Cuts with NME state}
Using NME states for a wire cut requires an extra Hilbert space $B$, with the sender controlling $A\otimes B$ and the receiver $C$.
The shared NME state $\rho_{BC}$ enables the use of a non-local operator $\mathcal{V}_{ABC}$ as depicted in \Cref{fig:wire_cut_model}~(c).
Analogously to \Cref{eq:id_without_NME}, $\mathcal{V}_{ABC}$ mimics the non-local identity $\mathcal{I}_{A \rightarrow C}$ between $A$ and $C$ when the sender's system $A\otimes B$ is traced out~\cite{Bechtold2024}:
\begin{align}\label{eq:non_local_V_with_NME}
        \forall \varphi \in D(A):\operatorname{Tr}_{AB}[\mathcal{V}_{ABC}(\varphi\otimes \rho_{BC})] = \varphi.
\end{align}
Similarly to \Cref{eq:qpd_without_NME}, wire cutting with NME state $\rho$ simulates the identity operator $\mathcal{I}_{A \rightarrow C}$ quasiprobabilistically.
A QPD for the non-local operator $\mathcal{V}_{ABC}$, given by $\mathcal{V}_{ABC}=\sum_i c_i \mathcal{F}_i$ with $\mathcal{F}_i \in \locc(A\otimes B,C)$,  is used.
Consequently, similar to \Cref{eq:min_sampling_overhead}, the optimal sampling overhead can be defined as
\begin{equation}\label{eq:min_sampling_overhead_wire_nme}
    \begin{split}
        \gamma^{\rho}(\mathcal{I}) := \min\biggl\{\sum_i |c_i|  \,\bigg|\,&  \mathcal{V} = \sum_i c_i \mathcal{F}_i,\, c_i\in\mathbb{R},\\
        &\mathcal{F}_i \in \locc(A\otimes B, C), \\
        &\forall \varphi \in D(A):\\
        &\operatorname{Tr}_{AB}[\mathcal{V}_{ABC}(\varphi\otimes \rho_{BC})] = \varphi\biggr\},
    \end{split}
    \raisetag{81pt}
\end{equation}
which represents the minimal sampling overhead of all QPDs for $\mathcal{V}_{ABC}$ that satisfy \Cref{eq:non_local_V_with_NME}~\cite{Bechtold2024}.
Using the fidelity of distillation~\cite{Regula2020}, another entanglement monotone, previous work established that the optimal sampling overhead for cutting a single wire with NME state is $\gamma^{\rho}(\mathcal{I})=\frac{4}{R(\rho)+1} -1$~\cite{Bechtold2023}.
Thus, the optimal sampling overhead depends on the entanglement in $\rho$.
An increase in the robustness value $R(\rho)$ leads to a decrease in the sampling overhead.
The generalization to the optimal sampling overhead for parallel cuts of multiple wires with NME states is later established in this work.

\begin{figure*}[ht]
    \centering
    \includegraphics{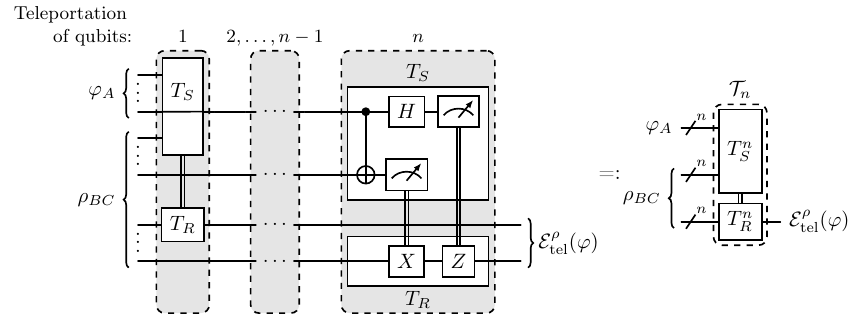}
    \caption{Teleportation circuit for $n$ qubits constructed from $n$ single-qubit teleportation protocols (left). Each single-qubit teleportation is surrounded by a gray box and involves the sender's ($T_S$) and receiver's ($T_R$) operations, shown in detail for the $n$-th teleportation. A condensed notation of the $n$-qubit teleportation is depicted on the right.}
    \label{fig:n_qubit_teleportation}
\end{figure*}

\subsubsection{An Optimal QPD without NME states}
A QPD for cutting $n$ parallel wires, achieving the optimal sampling overhead without any ancilla qubits, is provided by Harada et al.~\cite{Harada2023}:
\begin{equation}\label{eq:qpd_harada}
    \begin{split}
    \mathcal{I}^{\otimes n}(\bullet) =&\sum_{j=0}^{2^n-1} \sum_{l=0}^{2^n-1} \tr\left[ U_j\ketbras{l}U_j^{\dagger}(\bullet) \right] U_j\ketbras{l}U_j^{\dagger}\\
    &- \left(2^n-1\right)\sum_{l=0}^{2^n-1} \tr\left[ \ketbras{l}(\bullet) \right]\rho_l,
    \end{split}
    \raisetag{21pt}
\end{equation}
where 
\begin{align}
    \rho_l = \sum_{k=1}^{2^n-1}\Pr(k)\ketbras{l \oplus k}
\end{align}
and the probability of selecting $k$ from $0<k<2^n$ is uniformly distributed as:
\begin{equation}
\Pr(k) := \frac{1}{2^n - 1}.
\end{equation}
The unitary operators in the set $\{U_i\}_{i=0}^{2^n-1} \cup \{I^{\otimes n}\}$ are exactly the transformations that transform the computational basis into the complete set of $2^n+1$ different MUBs.
A specific construction for these unitary operators is detailed in \Cref{eq:mub_basis_transformation}.
However, this QPD is not limited to this specific construction and applies to any set of unitary operators that transform the computational basis into a complete set of MUBs~\cite{Harada2023}.
\Cref{fig:wire_cut_harada} depicts the circuits realizing this QPD.
Each sum over $l$ from \Cref{eq:qpd_harada} encompasses all possible measurement outcomes and is therefore implemented by a corresponding measurement in the circuits, resulting in a total of $2^n +1$ different circuits.

\subsection{Teleportation of Multi-Qubit States}\label{sec:teleportation}
In contrast to the quasiprobabilistic simulation of the transmission of a qubit's state, quantum teleportation actually transfers the state utilizing entanglement~\cite{Bennett1993}.
This protocol uses local quantum operations and communicates only two classical bits to transmit a qubit's state.
An error-free state transfer of a qubit requires a shared pair of maximally entangled qubits $\Phi_1$ between the sender and receiver~\cite{Prakash2012}.

To accurately teleport an $n$-qubit state without errors, a maximally entangled resource state of $2n$ qubits, i.e. $\Phi_n$, is required.
The teleportation can be optimally conducted through $n$ individual single-qubit teleportations~\cite{Chen2006}, each consuming a maximally entangled qubit pair $\Phi_1$ of $\Phi_n$.
This construction of the $n$-qubit teleportation circuit is depicted on the left side of \Cref{fig:n_qubit_teleportation}, using an arbitrary $2n$-qubit resource state $\rho_{BC}$, which must equal $\Phi_n$ to achieve optimal teleportation.
Individual single-qubit teleportations are represented by their condensed sender and receiver operations, $T_S$ and $T_R$, respectively, with the detailed implementation shown for the final teleportation procedure.
The right side of \Cref{fig:n_qubit_teleportation} shows a compact representation of the constructed $n$-qubit teleportation scheme, consolidating all sender and receiver operations into the transformations $T_S^n$ and $T_R^n$, respectively.
The complete $n$-qubit teleportation protocol is encapsulated by the operator $\mathcal{T}_n$, which operates on the $n$-qubit state $\varphi_A$ to be teleported and the $2n$-qubit resource state $\rho_{BC}$, represented as $\mathcal{T}_n(\varphi_A \otimes \rho_{BC})$.

When employing quantum teleportation with an arbitrary resource state $\rho_{BC}$, which may not be maximally entangled, errors occur and the resulting state is described according to Gu et al.~\cite{Gu2004} by
\begin{align}\label{eq:tele}
    \mathcal{E}_{\text{tel}}^{\rho}(\varphi) &= \operatorname{Tr}_{AB}\left[\mathcal{T}_n(\varphi_A \otimes \rho_{BC})\right] \\
    &=\sum_{\sigma \in \{I, X, Y, Z\}^{\otimes n}}\braket{\Phi^{\sigma}|\rho|\Phi^{\sigma}}\sigma\varphi\sigma.\label{eq:tele2}
\end{align}
Here, $\ket{\Phi^{\sigma}}$ represents the generalized Bell basis states associated with the $n$-qubit Pauli operators  $\sigma = \bigotimes_{i=0}^{n-1} \sigma_i$ with $\sigma_i \in \{I, X, Y, Z\}$. 
These generalized Bell states are defined as
\begin{align}\label{eq:phi_sigma}
    \ket{\Phi^{\sigma}} &= (\sigma \otimes I^{\otimes n})\ket{\Phi_n} \\
                        &= \bigotimes_{i=0}^{n-1} (\sigma_i \otimes I)\ket{\Phi_{1}}.
\end{align}
As a result, the occurrence of Pauli errors during teleportation depends on the overlap $\braket{\Phi^{\sigma}|\rho|\Phi^{\sigma}}$ between the resource state $\rho$ and the various generalized $n$-qubit Bell states $\Phi^{\sigma}$.

\section{Parallel Wire Cutting with NME states}\label{sec:main}
This section presents our results on how NME states can be used for joint cuts of multiple wires.
First, we introduce the optimal sampling overhead for parallel wire cuts when NME states are employed and examine the consequences of utilizing composite NME states, which are constructed from the combination of smaller NME states.
Although the results of the optimal sampling overhead are formulated in the framework of parallel wire cuts, we argue in \Cref{sec:discussion_arb_cuts} that they are also applicable to wire cuts at arbitrary circuit positions.
Following this, we propose an optimal QPD specifically designed for parallel wire cuts using pure NME states, and we demonstrate its effectiveness even when the available NME state is of a smaller dimension than that required for the wire cut.

\subsection{Optimal Sampling Overhead for a Parallel Wire Cut with NME states}\label{sec:parallel_overhead}
In scenarios where NME states are not employed in a wire cut, the sampling overhead required for the joint cutting of $n$ wires via classical communication aligns with the overhead observed in the simulation of a $2n$-qubit maximally entangled state~\cite{Brenner2023}.
Moreover, as discussed in \Cref{sec:quasi_prob_sim_states}, the sampling overhead for quasiprobabilistic simulation of maximally entangled states decreases when NME states are utilized compared to separable states~\cite{Takagi2024}. 
Building upon these insights, our prior research established that the sampling overhead for cutting a single wire reduces when using NME states, and it is consistent with the overhead for simulating a pair of maximally entangled qubits with an NME state~\cite{Bechtold2024}.
The following theorem extends these findings to the parallel cut of $n$ wires using an NME state. 

\begin{theorem}\label{theorem_overhead}
The optimal sampling overhead for a parallel wire cut of an $n$-qubit identity operator $\mathcal{I}^{\otimes n}$ using classical communication and an arbitrary $2n$-qubit NME resource state described by density operator $\rho$ is  
\begin{equation}
   \gamma^{\rho}(\mathcal{I}^{\otimes n}) = \frac{2^{n+1}}{R(\rho)+1} - 1.
\end{equation}
\end{theorem}
The detailed proof that reduces the optimal sampling overhead $\gamma^{\rho}(\mathcal{I}^{\otimes n})$ of a parallel cut of $n$ wires to the sampling overhead $\hat{\gamma}^{\rho}(\Phi_n)$ of simulating the maximally entangled $2n$ qubit state is provided in \Cref{sec:proof_overhead}.

The NME state $\rho$ discussed in \Cref{theorem_overhead} represents an arbitrary bipartite $2n$-qubit state.
However, future efforts in entanglement generation between quantum devices are likely to prioritize the creation of smaller entangled states, e.g. qubit pairs~\cite{Zhong2019,Kurpiers2018,Yan2022}, rather than focus on generating whole $2n$-qubit states.
This raises a crucial question: should one utilize a composite state constructed from individual shared entangled states for a single joint wire cut or perform separate cuts for each shared entangled state?
To address this question, we analyze the sampling overhead in the scenario where the objective is to cut $n$ wires with the aid of $m$ NME states that comprise a total of $2n$ qubits.
In this regard, the following theorem shows the benefit of utilizing the composite state in a single parallel cut:
\begin{theorem}\label{theorem_cut_composite_states}
    Let $\{\rho^{(i)}_{A_iB_i}\}_{i=0}^{m-1}$ be a set of $m \ge 2$ density operators with each $\rho^{(i)}_{A_iB_i}\in D(A_i \otimes B_i)$ for the corresponding $2^{k_i}$-dimensional Hilbert spaces $A_i$ and $B_i$ such that $n=\sum_{i=0}^{m-1} k_i$.
    Moreover, let each state $\rho^{(i)}_{A_iB_i}$ be an NME state, i.e., $0 \le R\left(\rho^{(i)}_{A_iB_i}\right) < 2^{k_i} -1$.
    Then, the following inequality holds for the composite state $\rho= \bigotimes_{i=0}^{m-1}\rho^{(i)}_{A_iB_i}$:
    \begin{align}
        \prod_{i=0}^{m-1}\gamma^{\rho^{(i)}_{A_iB_i}}(\mathcal{I}^{\otimes k_i}) > \gamma^{\rho}(\mathcal{I}^{\otimes n}). 
    \end{align}
\end{theorem}
The proof of \Cref{theorem_cut_composite_states} can be found in \Cref{sec:appendix_proof_cut_composite_states}. 
This theorem demonstrates that performing a parallel cut always results in a lower sampling overhead, even when the entangled resource state is a composite state constructed from smaller NME states. 
However, \Cref{theorem_cut_composite_states} does not address scenarios involving maximally entangled states within the composite state. 
This particular situation is the focus of the subsequent theorem.
\begin{theorem}\label{theorem_composite_max_ent}
    Consider an arbitrary bipartite state $\rho_{A_0B_0}\in D(A_0\otimes B_0)$ and a maximally entangled state $\Phi_{A_1B_1}\in D(A_1 \otimes B_1)$, where $A_0$ and $B_0$ are $2^{n_0}$-dimensional Hilbert spaces and $A_1$ and $B_1$ are $2^{n_1}$-dimensional Hilbert spaces. 
    For the sampling overhead of cutting $n$ wires, where $n = n_0 + n_1$, it holds that
    \begin{align}
        \nonumber&\gamma^{\rho_{A_0B_0} \otimes \Phi_{A_1B_1}}(\mathcal{I}^{\otimes n}) \\
        &= \gamma^{\rho_{A_0B_0}}(\mathcal{I}^{\otimes n_0})\gamma^{\Phi_{A_1B_1}}(\mathcal{I}^{\otimes n_1}) \\
        &= \gamma^{\rho_{A_0B_0}}(\mathcal{I}^{\otimes n_0})
    \end{align}
\end{theorem}
The proof of \Cref{theorem_composite_max_ent} is provided in \Cref{sec:appendix_proof_cut_composite_states}.
This theorem demonstrates that employing a single parallel cut using the composite state $\rho_{A_0B_0} \otimes \Phi_{A_1B_1}$, which includes a maximally entangled state, yields no advantage in sampling overhead compared to performing separate cuts for the two states.
Thus, maximally entangled states can be treated independently of the remaining NME states in a composite state.
The optimal sampling overhead, in this case, can be achieved by employing $\Phi_{A_1B_1}$ for quantum teleportations of $n_1$ wires and independently performing a joint cut on the remaining $n_0$ wires using $\rho_{A_0B_0}$.

In contrast to the scenario with maximally entangled states of \Cref{theorem_composite_max_ent}, \Cref{theorem_cut_composite_states} allows for parts of the composite resource state to be separable while still achieving a reduction in the sampling overhead.
Consequently, when the entangled resource state is smaller than required for cutting $n$ wires, the sampling overhead of a parallel cut of all $n$ wires, where the entangled resource state is augmented with a separable state, is lower than performing one cut with the NME state and a separate cut without entanglement for the remaining wires.
To exactly quantify the advantage in sampling overhead under these circumstances, consider the following theorem.
\begin{theorem}\label{theorem_composite_separable}
    Consider an entangled state $\rho_{A_eB_e}\in D(A_e\otimes B_e)$, where $A_e$ and $B_e$ are $2^{n_e}$-dimensional Hilbert spaces. 
    To cut $n$ wires, with $n>n_e$, consider the composite state $\rho_{A_eB_e} \otimes \tau_{A_sB_s}$, which augments the entangled state $\rho_{A_eB_e}$ with a separable state $\tau_{A_sB_s} \in S(A_s,B_s)$. 
    Here, $A_s$ and $B_s$ are $2^{n_s}$-dimensional Hilbert spaces, where $n_s= n-n_e$.
    The sampling overhead of a single parallel cut using the composite state $\rho_{A_eB_e} \otimes \tau_{A_sB_s}$ is lower than the overhead resulting from a parallel cut of $n_e$ wires with $\rho_{A_eB_e}$ and separately cutting the remaining $n_s$ wires without entanglement.
    The advantage in the sampling overhead is given by
    \begin{align}
        \nonumber&\gamma^{\rho_{A_eB_e}}(\mathcal{I}^{\otimes n_{e}})\gamma(\mathcal{I}^{\otimes n_{s}}) - \gamma^{\rho_{A_eB_e}\otimes \tau_{A_sB_s}}(\mathcal{I}^{\otimes n})\\
        &= \left(\gamma^{\rho_{A_eB_e}}(\mathcal{I}^{\otimes n_{e}}) - 1\right)(2^{n_s} - 1) \ge 0.\label{eq:theorem_composite_separable}
    \end{align}
\end{theorem}
The proof of \Cref{theorem_composite_separable} is provided in \Cref{sec:appendix_proof_cut_composite_states}. 
In contrast to \Cref{theorem_cut_composite_states}, which only establishes the benefit of using the composite resource state $\rho_{A_eB_e}\otimes \tau_{A_sB_s}$ in terms of the sampling overhead over cutting with $\rho_{A_eB_e}$ and applying a standard cut without entanglement separately, this theorem quantifies the precise advantage in this specific scenario.
The advantage of the single parallel cut $\gamma^{\rho_{A_eB_e}\otimes \tau_{A_sB_s}}(\mathcal{I}^{\otimes n})$ increases as the entanglement in $\rho_{A_eB_e}$ decreases, which is reflected by a higher sampling overhead $\gamma^{\rho_{A_eB_e}}(\mathcal{I}^{\otimes n_{e}})$ in \Cref{eq:theorem_composite_separable}.
The benefit of a lower sampling overhead achieved by employing a single parallel cut over two separate cuts diminishes entirely when $\rho_{A_eB_e}$ is maximally entangled, as in this case, $\gamma^{\rho_{A_eB_e}}(\mathcal{I}^{\otimes n_{e}}) - 1 = 0$.
This is consistent with \Cref{theorem_composite_max_ent}.

\begin{figure*}
    \centering
    \includegraphics{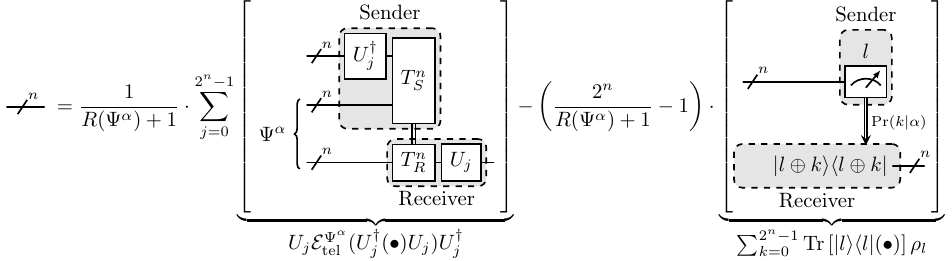}
    \caption{Parallel wire cut with NME state $\ket{\Psi^\alpha}$.}
    \label{fig:wire_cut_with_entanglement}
\end{figure*}

\subsection{Optimal QPD for a Parallel Wire Cut with Pure NME states}\label{sec:wire_cut_nme}
So far, the focus of this work has primarily been on understanding and quantifying the sampling overhead of a parallel cut of multiple wires in a quantum circuit using NME states.
This section presents a QPD that achieves this optimal sampling overhead to cut $n$ wires in parallel using pure NME states.

\begin{theorem}\label{theorem_decomposition}
    Let $\ket{\Psi^{\alpha}}$ denote a pure NME state of $2n$ qubits with entanglement robustness $R(\Psi^{\alpha})< 2^n -1$. 
The $n$-qubit identity operator $\mathcal{I}^{\otimes n}$ can be decomposed using quantum teleportation described by $\mathcal{E}_{\text{tel}}^{\Psi^{\alpha}}$ with $\Psi^{\alpha}$ as resource state as follows:
\begin{equation}\label{eq:thrm2}
    \begin{split}
    \mathcal{I}^{\otimes n}(\bullet) =&\phantom{-} \frac{1}{R(\Psi^{\alpha})+1}\sum_{j=0}^{2^n-1} U_j\mathcal{E}_{\text{tel}}^{\Psi^{\alpha}}\left(U_j^{\dagger}(\bullet) U_j\right)U_j^{\dagger}\\
    &- \left(\frac{2^n}{R(\Psi^{\alpha})+1}-1\right)\sum_{l=0}^{2^n-1} \tr\left[ \ketbras{l}(\bullet) \right]\rho_l,
    \end{split}
\end{equation}
where $U_j$ are the unitary operators from \Cref{eq:mub_basis_transformation} that transform the computational basis in the different MUBs and the density operator $\rho_l$ is 
\begin{align}
    \rho_l := \sum_{k=1}^{2^n-1}\Pr(k|\alpha)\ketbras{l \oplus k}
\end{align}
with the probabilities $\Pr(k|\alpha)$ for selecting $0<k<2^n$, calculated from the $2^n$-dimensional Schmidt vector $\alpha$ as 
\begin{align}
    \Pr(k|\alpha):=\frac{\left(\sum_{j=0}^{2^n-1}(-1)^{k \odot j}\alpha_j\right)^2}{2^n-1 -R(\Psi^{\alpha})}.
\end{align}
This decomposition achieves the optimal sampling overhead from \Cref{theorem_overhead}.
\end{theorem}

The detailed proof of this theorem can be found in \Cref{sec:proof_decomposition}, and
the operator $\odot$ denotes the multiplication in the Galois field $\mathbb{GF}(2^n)$ of dimension $2^n$ as detailed in \Cref{sec:galois_fields}.
The QPD of \Cref{theorem_decomposition} is depicted in \Cref{fig:wire_cut_with_entanglement}.
It consists of $2^n$ quantum teleportations with the pure NME state $\Psi^{\alpha}$, that apply the MUB transformations from $\{U_j\}_{j=0}^{2^n-1}$.
Additionally, the QPD includes one measure-and-prepare operation, i.e., the subtraction in \Cref{eq:thrm2}, that initializes a diagonal density operator $\rho_l$. 
The probability $\Pr(k|\alpha)$, determining the initialization of the state $\ket{l \otimes k}$ after the prior measurement of $\ket{l}$, is dependent on the Schmidt coefficients of $\Psi^{\alpha}$.

Since only the $2^n$ teleportation circuits require the entangled states $\Psi^{\alpha}$, the total number of entangled states consumed by sampling from the QPD is proportional to $2^n(R(\Psi^{\alpha})+1)^{-1}$. 
Therefore, the higher the entanglement in $\Psi^{\alpha}$, quantified by $R(\Psi^{\alpha})$, the fewer entangled states are needed to achieve the desired result accuracy. 
On the other hand, to consume an entangled state $\Psi^{\alpha}$, each teleportation circuit has $n$ additional ancilla qubits compared to the measure-and-prepare circuit on the right of \Cref{fig:wire_cut_with_entanglement}.
This circuit solely measures $n$ qubits of the sender and initializes $n$ qubits of the receiver. 
Consequently, the number of shots from circuits without ancilla qubits is proportional to $2^n(R(\Psi^{\alpha})+1)^{-1} -1$. 
This stands in contrast to an implementation using a QPD for the maximally entangled state, as described in \Cref{sec:quasi_prob_sim_states}, followed by a teleportation-based wire cut~\cite{Brenner2023}, where every circuit in the QPD requires ancilla qubits.

\begin{figure*}
    \centering
    \includegraphics{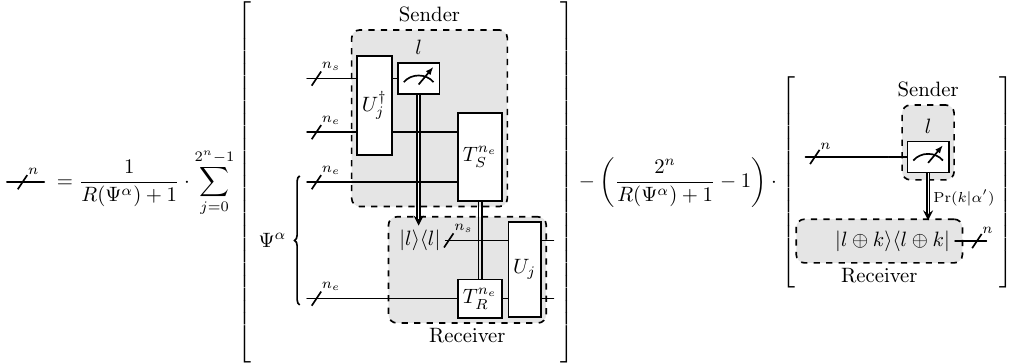}
    \caption{Parallel wire cut using NME state $\ket{\Psi^\alpha}$ consisting of $2n_e$ qubits to cut $n$ wires, where $n_e<n$, without ancilla qubits for separable states.}
    \label{fig:wire_cut_with_entanglement_sep}
\end{figure*}

\subsection{Optimal QPD with Pure States for \Cref{theorem_composite_separable}}\label{sec:QPD_k_pure_states}

Furthermore, in the scenario described in \Cref{theorem_composite_separable} where a pure NME state $\Psi^{\alpha}$, consisting of $2n_e$ qubits, is utilized to perform a wire cut on $n$ wires with $n_e < n$, the QPD of \Cref{theorem_decomposition} can still be applied. 
In this case, the NME state $\Psi^{\alpha}$ is supplemented by separable states for the additional $n_s=n-n_e$ qubits, specifically $(\ketbras{00})^{\otimes n_s}\otimes \Psi^{\alpha}$, which is then employed in the QPD. 

However, the process can be further streamlined. 
As the $n$-qubit teleportation is constructed from independent single qubit teleportation, we consider such an individual teleportation with a separable state.
The state $\mathcal{E}_{\text{tel}}^{\ketbras{00}}(\varphi)$, resulting from the teleportation of a one-qubit state $\varphi$ using the separable state $\ketbras{00}$, can be reduced to a simple measurement and reinitialization procedure, which is derived from \Cref{eq:tele2} in the following:
\begin{align}
    &\nonumber\mathcal{E}_{\text{tel}}^{\ketbras{00}}(\varphi)\\
    &=\sum_{\sigma \in \{I, X, Y, Z\}}\braket{\Phi^{\sigma}\ketbras{00}\Phi^{\sigma}}\sigma\varphi\sigma \\
    &=\sum_{\sigma \in \{I, X, Y, Z\}}|\!\braket{\Phi^{\sigma}|00}\!|^2\sigma\varphi\sigma \\
    &=\sum_{\sigma \in \{I, X, Y, Z\}}|\!\braket{\Phi|\sigma \otimes I|00}\!|^2\sigma\varphi\sigma,
\end{align}
where we applied the definition of $\ket{\Phi^{\sigma}}$ from \Cref{eq:phi_sigma} in the last step.
As a result of \Cref{eq:real_sum} in \Cref{sec:proof_decomposition}, it holds that
\begin{align}
    |\!\braket{\Phi|\sigma \otimes I|00}\!|^2 &= \frac{1}{2}\left(\braket{0|\sigma|0}\right)^{2}\\ &= \begin{cases}
        \frac{1}{2}, & \sigma \in \{I, Z\}\\
        0, & \sigma \in \{X, Y\}
    \end{cases}.
\end{align}
Therefore, we obtain
\begin{align}
    &\nonumber\mathcal{E}_{\text{tel}}^{\ketbras{00}}(\varphi)\\
    &= \frac{I\varphi I}{2} + \frac{Z\varphi Z}{2} \\
     \begin{split}
        &=\frac{(\ketbras{0} + \ketbras{1})\varphi(\ketbras{0}  + \ketbras{1})}{2}\\
        &\phantom{=} + \frac{(\ketbras{0} - \ketbras{1})\varphi (\ketbras{0} - \ketbras{1})}{2}
    \end{split} \\
    &= \ketbras{0}\varphi\ketbras{0} + \ketbras{1}\varphi\ketbras{1}\\
    &= \tr[\ketbras{0}\varphi]\ketbras{0} + \tr[\ketbras{1}\varphi]\ketbras{1} \\
    &= \sum_{j\in\{0,1\}} \tr[\ketbras{j}\varphi]\ketbras{j}\label{eq:measure_and_prepare}.
\end{align}
This simplification eliminates the need for inputting separable states $\ketbras{00}$.
All single-qubit teleportations of the $n$-qubit teleportation using a separable state can be replaced by the measure-and-prepare operations derived in \Cref{eq:measure_and_prepare}.
This streamlines the $n$-qubit teleportations circuits of the QPD from \Cref{theorem_decomposition} to be effectively executed with only $n_e$ ancilla qubits for the NME state $\Psi^{\alpha}$, as depicted in \Cref{fig:wire_cut_with_entanglement_sep}.
Consequently, the total number of qubits involved in the cutting procedure is reduced from $3n$ for the cut in \Cref{fig:wire_cut_with_entanglement} to $2n+n_e$ for the streamlined cut in \Cref{fig:wire_cut_with_entanglement_sep}.

The coefficients of the QPD are derived by leveraging that the robustness of entanglement remains invariant when a separable state is appended as shown in \Cref{lemma_robustness_composite_sep} in \Cref{sec:appendix_proof_cut_composite_states}, i.e., $R((\ketbras{00})^{\otimes (n_s)}\otimes \Psi^{\alpha}) = R(\Psi^{\alpha})$.
To account for the smaller NME state's dimensionality, the probability distribution $\Pr(k|\alpha)$ must be replaced by $\Pr(k | \alpha')$, where $\alpha'$ denotes the $2^n$ Schmidt coefficients of the expanded state $(\ketbras{00})^{\otimes (n_s)}\otimes \Psi^{\alpha}$. 
This adjustment ensures that the QPD accurately reflects the modified entanglement structure when using a smaller NME state.

\section{Discussion}\label{sec:discussion}

\begin{table}[t]
\centering
\begin{tabular}{l||c|c}
  & \thead{LOCC without\\NME states ($\gamma$)\\\cite{Brenner2023}} & \thead{LOCC with NME\\states $\rho$ ($\gamma^{\rho}$)} \\
 \hline\hline
 \thead{Single wire\\($\mathcal{I}$)} & \makecell{$3$} &  \makecell{$\displaystyle\frac{4}{R(\rho)+1}-1$\\\cite{Bechtold2024}}   \\
 \hline
 \thead{$n$ wires\\($\mathcal{I}^{\otimes n}$)} & \makecell{$2^{n+1} - 1$} &  \makecell{$\displaystyle\frac{2^{n+1}}{R(\rho)+1} - 1$\\{[\Cref{theorem_overhead}]}}
\end{tabular}
\caption{Comparison of the sampling overheads of the different wire cutting procedures.}
\label{tab:sampling_overheads}
\end{table}

\Cref{theorem_overhead} establishes the optimal sampling overhead for wire cutting when both classical communication and NME states are employed. 
This result extends the previously known optimal sampling overhead for parallel wire cutting, which only considered classical communication and did not incorporate the use of NME states~\cite{Brenner2023}. 
\Cref{tab:sampling_overheads} presents an overview of the optimal sampling overheads for wire cutting with and without NME states. 
Utilizing a separable resource state $\rho$ with $R(\rho)=0$ in a wire cutting procedure equates the sampling overheads in both scenarios.

In the following, \Cref{sec:discussion_arb_cuts} discusses how our findings on the optimal sampling overhead for parallel wire cuts using NME states can be applied to the sampling overhead for joint cuts at arbitrary circuit positions.
\Cref{sec:discussion_trade_off} examines the practical constraints associated with large-scale joint wire cuts employing NME states.  
Furthermore, \Cref{sec:discussion_harada} discusses how the QPD for parallel cuts with NME states, as outlined in \Cref{theorem_decomposition}, extends the work of Harada et al.~\cite{Harada2023} and examines the limitations of the QPD.

\subsection{Arbitrary Wire Cuts with NME States}\label{sec:discussion_arb_cuts}
The results for the sampling overhead in \Cref{sec:parallel_overhead} are formulated for parallel wire cuts.
A more general task is cutting $n$ wires at arbitrary positions in the circuit without the requirement that these cuts be aligned within the same time slice of a circuit~\cite{Brenner2023}.
This setting allows arbitrary operations between wire cuts.
Consequently, it is generally impossible to consolidate the identity operators of the cuts into the form $\mathcal{I}^{\otimes n}$.
Nevertheless, without using entanglement, both scenarios yield the same sampling overhead of $2^{n+1} - 1$~\cite{Brenner2023}, when the arbitrary operations between wire cuts are treated as black boxes and not considered in the cutting process~\cite{Schmitt2023}. 

Following the argumentation of Brenner et al.~\cite{Brenner2023}, the same can be shown when NME states are employed for cutting.
A parallel wire cut is a particular instance of the more general problem of cutting wires at arbitrary positions with black box operations in between.
Therefore, the sampling overhead $\gamma^{\rho}(\mathcal{I}^{\otimes n})$ for parallel wire cuts presented in this work provides a lower bound for the optimal sampling overhead for arbitrary wire cuts using NME states.

On the other hand, implementing $n$ wire cuts at arbitrary positions within a circuit can be achieved through $n$ teleportations.
These teleportations require $n$ maximally entangled qubit pairs, which %
are utilized at various points during the circuit's execution. %
By quasiprobabilistically simulating the $n$ required maximally entangled states $\ket{\Phi_1}^{\otimes n} = \ket{\Phi_n}$ together before the circuit's execution, one can achieve the sampling overhead $\hat{\gamma}^{\rho}(\Phi_n)$ when employing the NME state $\rho$. 
Therefore, $\hat{\gamma}^{\rho}(\Phi_n)$ serves as an upper bound for the sampling overhead of $n$ arbitrary wire cuts.
However, the sampling overhead $\gamma^{\rho}(\mathcal{I}^{\otimes n})$ for parallel cutting and the sampling overhead $\hat{\gamma}^{\rho}(\Phi_n)$ for the simulation of $\ket{\Phi_n}$ coincide (see \Cref{theorem_overhead} and \Cref{eq:optimal_overhead}).
Thus, the sampling overhead $\gamma^{\rho}(\mathcal{I}^{\otimes n})$ for parallel cutting as established in \Cref{theorem_overhead} also applies to the scenario of arbitrary wire cuts with black box operations in between. 
Consequently, \Cref{theorem_cut_composite_states,theorem_composite_max_ent,theorem_composite_separable} apply also to an arbitrary cut of $n$ wires with NME states.

When the operations between the wire cuts are considered in the cutting process, the sampling overhead can potentially be reduced by using a custom QPD for the overall operation, which includes both the identity operators to be cut and the intermediate operations.

\subsection{Trade-Off: Larger NME States and Resource Constraints}\label{sec:discussion_trade_off}
While \Cref{theorem_cut_composite_states} shows cutting all $n$ wires together has a lower sampling overhead than cutting them in smaller chunks, this strategy is limited in practice by the availability and size of (composite) entangled states.
The rate at which these states can be produced and the frequency with which consumed states can be regenerated may restrict their practical use~\cite{Pompili2022}. 
Assuming that the generation process itself is not a limiting factor, the potential size of a usable (composite) entangled state for a joint cut of $n$ wires scales linearly, equating to $2n$ qubits. 
However, from a practical point of view, generating arbitrarily large entangled states is not advantageous, as, e.g.,  decoherence rates increase with the size of the entangled state~\cite{Kam2023}. 
Moreover, such states would consume a significant portion of a quantum device's qubit resources.
This is counterproductive to the primary rationale for implementing circuit cutting techniques, i.e., to manage the limited qubit availability. 

Consequently, a more sensible approach,  similar to the concept of an entanglement factory as discussed by Pivetau et al.~\cite{Piveteau2023}, involves limiting the size of the generated (composite) entangled state to $2n_e$ qubits, sufficient for jointly cutting $n_e$ wires with the NME state.
The value of $n_e$ should balance between leveraging the benefits of larger entangled states, which can reduce the sampling overhead per wire cut, and ensuring that it remains within the practical limits of the hardware. 
Specifically, $n_e$ should be chosen so that it does not exceed $n$, aligns with entanglement generation capabilities, and utilizes only a reasonable fraction of a device's qubit resources.

To allow the cutting of more than $n_e$ wires using entanglement, this approach involves regenerating the entangled state as it is consumed. 
However, in situations where the regeneration of entanglement is not feasible, \Cref{theorem_composite_separable} proposes the alternative solution of supplementing the $2n_e$-qubit entangled state with separable states for the additional $n-n_e$ wires.
Importantly, whether separable states for parallel cuts need to be explicitly provided as inputs depends on the specific QPD.
\Cref{fig:wire_cut_with_entanglement_sep} shows that the QPD of \Cref{theorem_decomposition} can be modified not to require the input of separable states in the above scenario, and thus, they do not allocate any qubits.

\subsection{The QPD in \Cref{theorem_decomposition} as a Generalization and its Limitations}\label{sec:discussion_harada}
The optimal QPD with pure NME states, as presented in \Cref{theorem_decomposition}, is a generalization of the QPD without NME states outlined in \Cref{eq:qpd_harada} from Harada et al.~\cite{Harada2023}.
The strategy described in \Cref{sec:QPD_k_pure_states} demonstrates this generalization by allowing for the interpolation between the two scenarios by employing a pure NME state comprising $2n_e$ qubits, where $n_e<n$, for the purpose of cutting $n$ wires. 
Notably, the QPD from \Cref{eq:qpd_harada} without NME states is recovered in the special case where only separable states are used, corresponding to $n_e=0$.
In this case, the Schmidt vector $\alpha'= (1, 0, ..., 0)$ consists of a single non-zero entry, and thus the two probabilities $\Pr(k)$ and $\Pr(k|\alpha')$ coincide:
\begin{align}
    \Pr(k|\alpha') &= \frac{\left(\sum_{j=0}^{2^n-1}(-1)^{k \odot j}\alpha'_j\right)^2}{2^n-1 -R((\ketbras{00})^{\otimes n})}  \\
    &= \frac{1}{2^n-1} =  \Pr(k).
\end{align}

A drawback of the QPD described in \Cref{theorem_decomposition} is, however, that it necessitates the use of the specific MUB transformations $\{U_j\}_{j=0}^{2^n-1}$ as specified in \Cref{eq:mub_basis_transformation}.
In contrast, the original QPD by Harada et al.~\cite{Harada2023} is compatible with any set of $2^n$ operators that transform the computational basis into $2^n$ distinct MUBs.
To our knowledge, using such an arbitrary set of MUB transformations in conjunction with the teleportations with $\ket{\Psi^{\alpha}}$ does not allow completing the QPD by subtracting a single measure-and-prepare operation as in \Cref{eq:thrm2}.

Although the wire cut presented in \Cref{theorem_decomposition} is designed to support pure NME states of the form $\ket{\Psi^{\alpha}}$, its applicability extends to arbitrary bipartite pure states. 
As demonstrated in \Cref{eq:Psi_alpha}, any bipartite pure state $\ket{\psi}_{AB}$ can be transformed into the form of $\ket{\Psi^{\alpha}}$ by applying local unitary transformations to each partition, i.e., $\ket{\Psi^{\alpha}} = (U_A^{\dagger}\otimes U_B^{\dagger})\ket{\psi}_{AB}$.
For a given $\ket{\psi}_{AB}$, these unitary transformations can be determined once prior to the circuit execution by using the quantum singular value decomposer presented by Bravo-Prieto et al.~\cite{BravoPrieto2020}. 
Once obtained, these transformations can be reused for each wire cut as long as the resource state remains unchanged.

While the wire cut presented in \Cref{theorem_decomposition} is limited to pure NME states, it represents a significant step towards developing a wire cutting procedure for mixed states, which are more commonly encountered in practical applications due to the presence of noise during entanglement generation.
This work serves as a foundation for future research aimed at developing wire cutting techniques that are specifically tailored to mixed NME states resulting from the noise characteristics of the practical quantum hardware at hand.

\section{Related Work}\label{sec:related_work}
Simulating maximally entangled states using NME states through quasiprobability simulation, as discussed in \Cref{sec:quasi_prob_sim_states}, has been extensively studied in the context of virtual entanglement distillation~\cite{Takagi2024,Yuan2024}. 
This approach diverges from traditional entanglement distillation, which aims to transform an NME state into a physical state with higher entanglement~\cite{Bennett1996a,Bennett1996b,Rozpedek2018}. 
Virtual distillation focuses solely on replicating the measurement statistics of the target state through quasiprobability simulation. 
Its feasibility has been recently demonstrated in experiments~\cite{Zhang2023b}.
The resulting virtually entangled states can be employed in quantum protocols that rely on maximally entangled states.
Consequently, a wire cut can be realized by using these virtual states in quantum teleportation.
Our research, however, explores an alternative strategy by directly implementing a wire cut with NME states, effectively achieving a virtual distillation of the identity operator~\cite{Takagi2024}.
This method bypasses the preliminary step of simulating a maximally entangled state before teleportation, offering a more efficient pathway. 

In addition to scenarios where classical communication facilitates wire cutting, wire cutting can be accomplished solely through local operators~(LO) without classical communication, as demonstrated by Peng et al.~\cite{Peng2019}.
However, this approach incurs an increased sampling overhead of $\gamma_{LO}(\mathcal{I}) = 4$, and there is no advantage in simultaneously cutting multiple wires, i.e., $\gamma_{LO}(\mathcal{I}^{\otimes n}) = 4^n = \left(\gamma_{LO}(\mathcal{I})\right)^n$~\cite{Brenner2023}.
For observables of rank at most $r\le 2^n$, the sampling overhead above can be generalized to $\gamma_{LO}(\mathcal{I}^{\otimes n}) = 2^nr$~\cite{Harrow2024}.
A natural question to ask, but beyond the scope of this paper, is whether entanglement can lower the sampling overhead in this context.

Nevertheless, optimizing the sampling overhead is not the sole objective.
The number of circuits in the QPD is also a critical factor as each subcircuit requires additional compilation time and therefore increases the total computation time~\cite{Harada2023}.
In scenarios without entanglement, $2^n+1$ circuits represent the optimal configuration for a parallel cut of $n$ wires.
This raises the question of what is the minimal number of circuits in wire cutting with entangled states.
However, addressing this question is beyond the scope of this paper. 

In addition to wire cutting, gate cutting is an alternative circuit cutting technique~\cite{Piveteau2023,Mitarai2021,Ufrecht2023a}.
It encompasses cutting multi-qubit gates by quasiprobabilistically simulating the corresponding non-local operator with local operators as described in \Cref{sec:quasi_prob_sim}.
Similar to wire cutting, joint cutting of multiple gates has been shown to reduce the sampling overhead compared to individual gate cuts~\cite{Piveteau2023,Ufrecht2023b,Schmitt2023}.
The choice between wire cutting and gate cutting to achieve a lower sampling overhead depends on the specific circuit under consideration~\cite{Brenner2023}.  
Harrow et al.~\cite{Harrow2024} introduce a new measure of the entangling power of unitary operations called the product extent and a procedure for gate cutting whose sampling overhead equals the product extent, providing a direct link between the entanglement properties of the gates and the efficiency of the cutting process.
Furthermore, Jing et al.~\cite{Jing2024} establish an exponential lower bound on the sampling overhead assisted by local operations and classical communication for any non-local operator. 
This lower bound is determined by the entanglement cost of the operator under separability preserving operations.

Techniques to automatically identify locations for wire and gate cuts in quantum circuits have been introduced~\cite{Tang2021,Brandhofer2024,Kan2024}. 
Pawar et al.~\cite{Pawar2023} additionally integrate circuit cutting with the reuse of qubits to further reduce the circuit size.
Furthermore, to tackle statistical errors introduced by the finite sampling from the subcircuits in wire cutting, maximum likelihood methods have been proposed~\cite{Perlin2021}.
While the primary goal of circuit cutting is to reduce the size of quantum circuits, several empirical studies have shown that executing smaller subcircuits can lead to improved overall results~\cite{Perlin2021,Ayral2021,Bechtold2023}.

\section{Conclusion}\label{sec:conclusion}
This work demonstrates the utility of NME states in jointly cutting multiple wires in a quantum circuit.
We identify the optimal sampling overhead for this scenario, which generalizes the established sampling overhead without NME states.
The higher the entanglement utilized in the joint cut, the lower its sampling overhead.
This optimal sampling overhead can be achieved using pure NME states with the parallel wire cut introduced in \Cref{sec:wire_cut_nme}.

Moreover, combining smaller NME states into a single composite state for a single joint cut reduces the sampling overhead compared to performing individual cuts for each constituent NME state.
However, this advantage does not hold for maximally entangled components of a composite state, as their integration into the joint cut does not offer any additional reduction in the sampling overhead.
These maximally entangled components can be utilized independently in quantum teleportations.
Furthermore, when an NME state is insufficiently large to cut all targeted wires, it can be supplemented with a separable state to perform a single joint cut, resulting in a lower sampling overhead than performing one cut with the given NME state and a separate cut without entanglement for the remaining wires.
The results of this work bridge the gap between traditional wire cutting of multiple wires without entanglement and quantum teleportations with maximally entangled states by providing a flexible approach that leverages NME states.

Building upon the introduced wire cutting protocol that employs pure NME states, future work can investigate wire cutting protocols that utilize mixed NME states, enabling the utilization of entangled states that have experienced decoherence due to noise.
Additionally, applying the idea of utilizing NME states to gate cutting techniques to reduce their sampling overhead is another promising research avenue.

\section*{Acknowledgment}
This work was partially funded by the BMWK projects \textit{EniQmA} (01MQ22007B) and \textit{SeQuenC} (01MQ22009B).

\bibliographystyle{quantum}
\bibliography{bibliography-trunc5}

\onecolumn
\appendix

\section{Galois Fields}\label{sec:galois_fields}
The used construction of the MUBs from \Cref{sec:mubs} and the proof of the $n$-qubit wire cut with pure NME states relies on Galois fields.
Therefore, this section gives a brief introduction. %
Within the context of a $n$-qubit quantum system, we consider finite fields, commonly referred to as Galois fields, consisting of $2^n$ elements.
For $n=1$, the finite field is $\mathbb{F}_2 = \{0, 1\}$ with addition and multiplication modulo $2$.
For $n>1$, this Galois field is denoted as $\mathbb{GF}(2^n)$. 
It is entirely defined by its number of elements, as all finite fields with the same number of elements are equivalent up to the relabeling of the elements~\cite{Kibler2017}.
Typically, the Galois field $\mathbb{GF}(2^n)$ is constructed with the Galois extension of degree $n$ of the field $\mathbb{F}_2$~\cite{Kibler2017}.
To this end, we consider the ring of polynomials $\mathbb{F}_2[x]$ in the indeterminate $x$ with coefficients in the field $\mathbb{F}_2$.
Especially, a monic irreducible polynomial $P_n(x)$ of degree $n$ is employed, with coefficients $\mu_i \in \mathbb{F}_2$. 
This polynomial is expressed as:
\begin{align}
P_n(x) = \sum_{i=0}^{n} \mu_ix^i,
\end{align}
where monic means that $\mu_n=1$, and irreducible implies that the polynomial is non-constant and cannot be factored into the product of two non-constant polynomials in $\mathbb{F}_2[x]$. 
Note that such an irreducible polynomial always exists, and there are efficient algorithms for systematically constructing it~\cite{Shoup1994,Couveignes2012}.
Utilizing a monic irreducible polynomial $P_n(x)$, the field $\mathbb{GF}(2^n)$ is the residue class field of the ring of polynomials $\mathbb{F}_2[x]$ modulo $P_n(x)$, that is,
\begin{align}
    \mathbb{GF}(2^n) = \mathbb{F}_2[x]/\langle P_n(x) \rangle.
\end{align}

Consequently, each element $\mathbb{GF}(2^n)$ can be uniquely labeled by an $n$-tuple $\Vec{a}=(a_0, a_1, \ldots, a_{n-1})$ with $a_i \in \mathbb{F}_2$, defining a polynomial of degree less than $n$ as given by
\begin{align}\label{eq:galois_polynomial}
a(x) = \sum_{i=0}^{n-1} a_ix^i.
\end{align}
Equally, each element in $\mathbb{GF}(2^n)$ can be interpreted as an integer value between $0$ and $2^n-1$, that is given by
\begin{align}
a=\sum_{i=0}^{n-1} a_i 2^i.
\end{align}
Depending on the context, we will either use the $n$-tuple representation $\Vec{a}$ or the integer representation $a$ of the element in $\mathbb{GF}(2^n)$.

Furthermore, each field is characterized by two fundamental operations: an addition, denoted as $\oplus$, and a multiplication, denoted as $\odot$.
The addition and multiplication of the field elements in $\mathbb{GF}(2^n)$ are then carried out as addition and multiplication of the corresponding polynomials modulo the irreducible polynomial $P_n(x)$.
Therefore, the addition $\oplus$ is defined as
\begin{align}
a \oplus b &= a(x) + b(x) \mod P(x)\\
&=\sum_{i=0}^{n-1} (a_i \oplus b_i) x^i \\
&= (a_0 \oplus b_0, a_1 \oplus b_1, \ldots, a_{n-1} \oplus b_{n-1})
\end{align}
where $a_i+b_i$ is the addition in $\mathbb{F}_2$, and thus, the addition $a \oplus b$ is equivalent to the component-wise addition modulo $2$ of the $n$-tuples $a$ and $b$.
Moreover, the multiplication $\odot$ is defined as
\begin{align}
    a \odot b &= a(x) \cdot b(x) \mod P(x)\\
    &= \left(\sum_{s=0}^{n-1} a_s x^s\right)\left(\sum_{t=0}^{n-1} b_t x^t\right) \mod P(x)\\
    &= \sum_{s,t=0}^{n-1} a_s b_t x^{s+t} \mod P(x). \label{eq:mult_pol}
\end{align}

To derive an $n$-tuple $\vec{c}$, where $c = a \odot b$ corresponds to the result of \Cref{eq:mult_pol}, polynomial division by $P(x)$ is required when the multiplication process yields a polynomial with non-zero coefficients for terms $x^{s+t}$ where $s+t \ge n$.
The following lemma shows that the multiplication can be reduced to multiplication with $n$ matrices.
\begin{lemma}\label{lemma_mult_matrix_galois}
    For a Galois field $\mathbb{GF}(2^n)$, where multiplication is performed modulo an irreducible polynomial $P_n(x)$, there exists $n$ symmetric matrices $\mathcal{M}_i\in \mathbb{F}_2^{n \times n}$such that for any elements $a,b \in \mathbb{GF}(2^n)$, the tuple resulting from $a \odot b$ can be computed as
    \begin{align}
        a \odot b &= (\Vec{a}\mathcal{M}_0 \Vec{b}^T, \Vec{a}\mathcal{M}_1 \Vec{b}^T, \ldots , \Vec{a}\mathcal{M}_{n-1} \Vec{b}^T).
    \end{align}
    The matrices $\mathcal{M}_i$ are invertible in modulo two arithmetic.
    \begin{proof}
        The detailed construction of these matrices is provided in the reference~\cite{Durt2010}.
    \end{proof}
\end{lemma}

Furthermore, when an element of $\mathbb{GF}(2^n)$ is used as an exponent, its computation is based on its integer representation. 
When the base is $-1$ with exponent $a\in \mathbb{GF}(2^n)$, the result is determined solely by $a_0$, that is,
\begin{align}\label{eq:neg_one_exp_a_0}
    (-1)^a = (-1)^{\sum_{i=0}^{n-1} a_i 2^i} = (-1)^{a_0}.
\end{align}
This implies the following identity for $a,b\in \mathbb{GF}(2^n)$, which we will use later:
\begin{align}\label{eq:gf_exponent_sum}
    (-1)^a(-1)^b = (-1)^{a_0}(-1)^{b_0} = (-1)^{a_0+b_0} = (-1)^{(a \oplus b)_0} = (-1)^{a \oplus b}.
\end{align}
Additionally, in the later proof of \Cref{theorem_decomposition} given in \Cref{sec:proof_decomposition}, we use the identity given by the following lemma.
\begin{lemma}\label{lemma_sum_galois}
    Let $a\in \mathbb{GF}(2^n)$, then it holds that
    \begin{align}
        \sum_{j=0}^{2^n-1}(-1)^{j \odot a} = 2^n\delta_{a,0}
    \end{align}
    where $\delta_{a,0}$ is the Kronecker delta.
    \begin{proof}
        The proof is given in Equation (2.17) of~\cite{Durt2010}.
    \end{proof}
\end{lemma}

\section{Constructing MUBs using Shift and Phase Operators}\label{sec:phase_shift_operator}
Specifying the maximal size of a set of MUBs for a fixed dimension $d$ is still an open question in quantum information theory~\cite{Horodecki2022}, which, among other things, is important for finding optimal schemes for orthogonal quantum measurement in this dimension.
However, it is known that the number of MUBs for a dimension $d$ is at most $d+1$, and if the dimension $d$ is a power of a prime $p$, i.e., $d = p^n$, it was shown that there exists a set of $d+1$ MUBs~\cite{Wootters1989}.
We refer to a set of MUBs that achieves this maximal size of $d+1$ as \emph{complete}.

From a complete set of mutually unbiased bases, unitary operators can be generated that play a crucial role in our wire cutting protocol.
Since unitary operators form a vector space, this space can be endowed with the Hilbert-Schmidt inner product, defined for unitary operators $A$ and $B$ as $\braket{A,B}_{HS}= \tr(A^{\dagger}B)$. 
Orthogonality between $A$ and $B$ is achieved when $\braket{A,B}_{HS} = 0$.
Given this, a complete set of MUBs $\{\mathcal{B}_j\}_{j=0}^{d}$ is fundamentally connected to a set of $d^2$ orthogonal unitary operators $\{S_{j,k}\}$~\cite{Bandyopadhyay2002}.
Specifically, $\mathcal{B}_j$ serves as the eigenbasis for all $S_{j,k}$ with $k \in \{0, \ldots, d-1\}$.
A characteristic of this set of operators $\{S_{j,k}\}$ is that the operators can be organized into $d+1$ subsets, each containing $d$ elements that commute pairwise~\cite{Bandyopadhyay2002}.
The only common element among all subsets is the identity operator, denoted by $S_{j,0}$ for the different subsets. 
In general, the operator $S_{j,k}$ is the $k$-th unitary operator in the $j$-th subset, for $j \in \{0, \ldots, d\}$ and $k \in \{0, \ldots, d-1\}$.
Since all $d$ operators $\{S_{j,k}\}_{k=0}^{d-1}$ in the $j$-th subset commute, they have a shared eigenbasis that simultaneously diagonalizes all of the operators~\cite[Theorem
1.3.21]{Horn1985}.
This eigenbasis coincides with the MUB $\mathcal{B}_j$~\cite{Bandyopadhyay2002}.
Consequently, each unitary operator $S_{j,k}$ with its eigenvalues $\lambda_{j,k,l}$ can be expressed as
\begin{align}\label{eq:sab_sum_of_mub}
    S_{j,k} = \sum_{l=0}^{d-1} \lambda_{j,k,l} \ketbras{e_{l}^{j}}.
\end{align}

Fortunately, when considering the Hilbert space of $n$ qubits, the dimension is $d=2^n$, which is a power of a prime. 
Consequently, for $n$ qubits, there exists a complete set of $2^n +1$ MUBs. 
To construct such a complete set of MUBs, we follow the approach from Durt et al.~\cite{Durt2010}, which defines the orthogonal unitary operators $\{S_{j,k}\}$. 
Hereby, basis $\mathcal{B}_{2^n}$ is chosen as the computational basis, i.e, $\mathcal{B}_{2^n}= \{\ket{l}\}_{l=0}^{2^n-1}$.
The corresponding commuting unitary operators for $j=2^n$ in this construction are thereby given by $S_{2^n,k}=\hat{Z}_k$ with the \emph{phase operator} $\hat{Z}_k$ defined as
\begin{align}
\hat{Z}_k &:= \sum_{l=0}^{2^n-1}(-1)^{l\odot k}\ketbras{l},\label{eq:phase_op}
\end{align}
where $\odot$ denotes the multiplication in the Galois field $\mathbb{GF}(2^n)$ of dimension $2^n$.
Here, $l$ and $k$ are interpreted as the integer representations of elements in this finite field, as detailed in \Cref{sec:galois_fields}. 
The remaining $2^n$ MUBs are constructed by completing the set of orthogonal unitary operators $S_{j,k}$.
Following Durt et al.~\cite{Durt2010}, the operators for $j,k \in \{0, \ldots, 2^n-1 \}$ are defined as
\begin{align}\label{eq:S_jk}
    S_{j,k} := s_{j,k}\hat{Z}_{j\odot k}\hat{X}_k
\end{align}
with phase factor 
\begin{align}\label{eq:phase_factor}
    s_{j,k} = \prod_{0 \le r,t \le n-1 } \iu^{j\odot(k_r2^r)\odot (k_t2^t)}
\end{align}
and the \emph{shift operator} $\hat{X}_k$, which is defined as
\begin{align}
\hat{X}_k &:= \sum_{l=0}^{2^n-1}\ket{l}\bra{l\oplus k}.
\end{align}
Here, $\oplus$ denotes the addition in $\mathbb{GF}(2^n)$ as detailed in \Cref{sec:galois_fields}. 
Moreover, for the dimension $2^n$, the shift and phase operator can be constructed via products of Pauli $Z$ and $X$ gates, respectively, as detailed below in \Cref{sec:shift_phase_op_as_paulis}.

The resulting joint eigenbasis of the $j$-th commuting subset of the constructed unitary operators is $\mathcal{B}_j=\{\ket{e^j_l}\}_{l=0}^{2^n-1}$ with the $l$-th element of the basis $\mathcal{B}_j$ is given by
\begin{align}
    \ket{e_l^j} = \frac{1}{\sqrt{2^n}}\sum_{k=0}^{2^n-1} (-1)^{l \odot k} \overline{s_{j,k}} \ket{k},
\end{align}
where 
\begin{align}
    \overline{s_{j,k}} = \prod_{0 \le r,t \le n-1 } (-\iu)^{j\odot(k_r2^r)\odot (k_t2^t)}
\end{align}
is the complex conjugate of the phase factor $s_{j,k}$~\cite{Durt2010}.
For this particular construction of MUBs~\cite{Durt2010}, \Cref{eq:sab_sum_of_mub} can be refined to
\begin{align}
    S_{j,k} = \sum_{l=0}^{2^n-1} (-1)^{l \odot k} \ketbras{e_{l}^{j}}.
\end{align}

\subsection{The Shift and Phase Operators as Products of Pauli Operators}\label{sec:shift_phase_op_as_paulis}
The shift and phase operators for $n$ qubits are defined with respect to the Galois field $\mathbb{GF}(2^n)$.
In this context, an orthonormal basis of the $2^n$-dimensional Hilbert space is indexed by the elements of $\mathbb{GF}(2^n)$, which means that each basis state $\ket{a}$ corresponds to a field element $a\in \mathbb{GF}(2^n)$. 

It holds for the shift operator that
\begin{align}
    \hat{X}_a &= \sum_{k=0}^{2^n-1}\ket{k}\bra{k\oplus a}\\
    &= \sum_{k=0}^{2^n-1} \bigotimes_{i=0}^{n-1} \ket{k_i}\bra{k_i\oplus a_i}\label{eq:separable_basis_states}\\
    &= \sum_{k_0 \in \{0,1\}} \ldots \sum_{k_{n-1}\in\{0,1\}} \bigotimes_{i=0}^{n-1} \ket{k_i}\bra{k_i\oplus a_i},
\end{align}
since $\ket{k}$ and $\ket{k \oplus a}$ are trivially separable using their binary decomposition given by $k = \sum_{i=0}^{n-1}(k_i2^i)$ and $k\oplus a = \sum_{i=0}^{n-1}(k_i2^i) \oplus (a_i2^i)$ with $k_i, a_i \in \mathbb{F}_2$.
We proceed by exchanging the order of the summation and the tensor product:
\begin{align}
    \hat{X}_a &= \bigotimes_{i=0}^{n-1} \sum_{k_{i}\in\{0,1\}} \ket{k_i}\bra{k_i\oplus a_i}\\
    &= \bigotimes_{i=0}^{n-1} \left(\ketbra{0}{0 \oplus a_i} + \ketbra{1}{1 \oplus a_i}\right)\\
    &= \bigotimes_{i=0}^{n-1} X^{a_i} 
\end{align}
where $X^{a_i}$ is the Pauli $X$ gate classically controlled by $a_i$. %

Furthermore, it holds for the phase operator that 
\begin{align}\label{eq:phase_op_pauli_prod}
    \hat{Z}_a &= \sum_{k=0}^{2^n-1}(-1)^{k\odot a}\ket{k}\bra{k}\\
    &= \sum_{k=0}^{2^n-1}(-1)^{\Vec{k} \mathcal{M}_0 \Vec{a}^T}\ket{k}\bra{k} \label{eq:insert_M_0}\\
    &= \sum_{k=0}^{2^n-1}(-1)^{\Vec{k} \mathcal{M}_0 \mathcal{M}_0^{-1} \Vec{a'}^T}\ket{k}\bra{k} \label{eq:M_inverse_a_prime}\\
    &= \sum_{k=0}^{2^n-1}(-1)^{\Vec{k}\Vec{a'}^T}\ket{k}\bra{k} \\
    &= \sum_{k=0}^{2^n-1} \bigotimes_{i=0}^{n-1} (-1)^{k_i a'_i}\ket{k_i}\bra{k_i} \\
    &= \bigotimes_{i=0}^{n-1} \left( \sum_{k_i\in \{0,1\}} (-1)^{k_i a'_i}\ket{k_i}\bra{k_i} \right) \\
    &= \bigotimes_{i=0}^{n-1} \left(\ketbras{0} + (-1)^{a'_i} \ketbras{1}\right)\\
    &= \bigotimes_{i=0}^{n-1} Z^{a'_i} 
\end{align}
where \Cref{eq:insert_M_0} uses $(-1)^{k\otimes a} = (-1)^{(k\otimes a)_0}$ resulting from \Cref{eq:neg_one_exp_a_0} and then applies $(k\otimes a)_0 = \Vec{k} \mathcal{M}_0 \Vec{a}^T$ resulting from \Cref{lemma_mult_matrix_galois}.
Moreover, \Cref{eq:M_inverse_a_prime} is valid because the matrix $\mathcal{M}_0$ from \Cref{lemma_mult_matrix_galois} is invertible, and thus there exists for each $a\in \mathbb{GF}(2^n)$ an unique $a' \in \mathbb{GF}(2^n)$ such that $\Vec{a}^T=\mathcal{M}_0^{-1} \Vec{a'}^T$.
Additionally, $Z^{a'_i}$ represents the Pauli $Z$ gate classically controlled by $a'_i$. %
Thus, the phase operator $\hat{Z}_a$ can be constructed by a product of Pauli $I$ and $Z$ gates.

\subsection{Commutation of the Shift and Phase Operators}
The following lemma describes the commutation relation of the phase and the shift operators.
\begin{lemma}\label{lemma_commutation_XZ}
    In general, the shift operator $\hat{X}_a$ and the phase operator $\hat{Z}_b$ for $a,b \in \mathbb{GF}(2^n)$ do not commute.
    The commutation relation is given by
    \begin{align}
        \hat{Z}_{a}\hat{X}_b = (-1)^{a\odot b} \hat{X}_b\hat{Z}_{a}.
    \end{align}
\end{lemma}
\begin{proof}
    \begin{align}
    \hat{Z}_{a}\hat{X}_b &= \left(\sum_{k=0}^{2^n-1}(-1)^{k\odot a}\ket{k}\bra{k}\right) \left(\sum_{l=0}^{2^n-1}\ket{l}\bra{l\oplus b}\right) \\
    &=\sum_{k=0}^{2^n-1}(-1)^{k\odot a}\ket{k}\bra{k\oplus b}.
    \end{align}
Notice that in $\mathbb{GF}(2^n)$ each element is self-inverse regarding the addition.
Thus, $k = k \oplus b \oplus b = k' \oplus b$ where $k'= k \oplus b$.
Therefore, we can replace the sum over $k$ with the sum over $k'$ since for each $k$, there exists a unique $k'$, and the sum is still over all elements of $\mathbb{GF}(2^n)$.
As a result, we obtain
\begin{align}
    \hat{Z}_{a}\hat{X}_b &=\sum_{k'=0}^{2^n-1}(-1)^{(k' \oplus b)\odot a}\ket{k'\oplus b}\bra{k'} \\
    &=\sum_{k'=0}^{2^n-1}(-1)^{b \odot a}(-1)^{k' \odot a}\ket{k'\oplus b}\bra{k'}\ket{k'}\bra{k'} \\
     &=(-1)^{a \odot b}\left( \sum_{l'=0}^{2^n-1}\ket{l' \oplus b}\bra{l'}\right) \left( \sum_{k'=0}^{2^n-1}(-1)^{k' \odot a}\ket{k'}\bra{k'} \right) \label{eq:double_sum} \\
    &=(-1)^{a \odot b}\left( \sum_{l=0}^{2^n-1}\ket{l}\bra{l\oplus b}\right) \left( \sum_{k'=0}^{2^n-1}(-1)^{k' \odot a}\ket{k'}\bra{k'} \right) \label{eq:replace_l_prime} \\
    &=(-1)^{a \odot b}\hat{X}_b\hat{Z}_{a}.
    \end{align}
\end{proof}
\Cref{eq:double_sum} is valid, since all for all $l'\ne k'$ the term $\braket{l'|k'}=0$ and therefore only terms with $l' = k'$ remain.
\Cref{eq:replace_l_prime} replaces again the sum over $l'$ with the sum over $l$ such that $l' \oplus b = l$.

\section{Proof of \Cref{theorem_overhead}}\label{sec:proof_overhead}
The theorem is restated for convenience.
\begingroup
\def\thetheorem{\ref{theorem_overhead}}
\begin{theorem}
The optimal sampling overhead for a parallel wire cut of an $n$-qubit identity operator $\mathcal{I}^{\otimes n}$ using classical communication and an arbitrary $2n$-qubit NME resource state described by density operator $\rho$ is  
\begin{equation}
   \gamma^{\rho}(\mathcal{I}^{\otimes n}) = \frac{2^{n+1}}{R(\rho)+1} - 1.
\end{equation}
\end{theorem}
\addtocounter{theorem}{-1}
\endgroup

\begin{proof}
    To establish the optimal sampling overhead for a parallel cut of $n$ wires using arbitrary $2n$-qubit states, we generalize the proof for cutting a single wire employing arbitrary two-qubit states, as outlined in~\cite{Bechtold2024}.
    Our goal is to establish that the sampling overhead $\gamma^{\rho}(\mathcal{I}^{\otimes n})$ for parallel cutting $n$ wires with arbitrary $2n$-qubit states $\rho$ is equivalent to the sampling overhead  $\hat{\gamma}^{\rho}(\Phi_n)$ for simulating the maximally entangled $2n$-qubit state $\Phi_n$ using $\rho$.
    This equivalence is demonstrated by proving that $\hat{\gamma}^{\rho}(\Phi_n)$ serves as both an upper and lower bound for $\gamma^{\rho}(\mathcal{I}^{\otimes n})$.

    The upper bound $\gamma^{\rho}(\mathcal{I}^{\otimes n}) \le \hat{\gamma}^{\rho}(\Phi_n)$ is established by showing that a QPD for $\Phi_n$ can be used to simulate the identity operator $\mathcal{I}^{\otimes n}$ using quantum teleportation. 
    Consider three $n$-qubit Hilbert spaces $A$, $B$, $C$.
    A wire cut based on teleportation can implement the identity operator $\mathcal{I}^{\otimes n}_{A \rightarrow C}$ to quasi-probabilistically transfer an $n$-qubit state $\varphi_A$ from $A$ to $C$.
    This process necessitates the maximally entangled $2n$-qubit state $\Phi_n$ %
    , which can be quasiprobabilistically simulated.
    To this end, a QPD for the maximally entangled state $\Phi_{BC}$ is employed, based on the NME state $\rho_{BC}$. 
    The QPD is expressed as $\Phi_{BC} = \sum_i c_i \mathcal{F}_i(\rho_{BC})$, where $\mathcal{F}_i \in \locc(B,C)$ and the sampling overhead is given by $\sum_i |c_i|$.
    Incorporating this QPD within the $n$-qubit teleportation described by the operator $\mathcal{T}_n$, yields a QPD for the non-local $n$-qubit identity operator $\mathcal{I}^{\otimes n}_{A \rightarrow C}$ between systems $A$ and $C$ when $A\otimes B$ is traced out:
    \begin{align}
        \forall \varphi \in D(A): \varphi &= \operatorname{Tr}_{AB}\left[\mathcal{T}_n(\varphi\otimes\Phi_{BC})\right]\\
        &= \sum_i c_i \operatorname{Tr}_{AB}\left[\mathcal{T}_n(\varphi \otimes \mathcal{F}_i(\rho_{BC}))\right].
    \end{align}
    This establishes an $n$-qubit wire cut using NME states $\rho_{BC}$, as each operator within the QPD solely consists of LOCC between $A\otimes B$ and $C$: $\mathcal{T}_n\in \locc(A\otimes B, C)$ and $\mathcal{F}_i \in \locc(B,C)$.
    The sampling overhead of this teleportations-based wire cut is $\sum_i |c_i|$, which, by design, is equivalent to the sampling overhead for simulating the maximally entangled $2n$-qubit state using $\rho$. 
    Consequently, the upper bound $\gamma^{\rho}(\mathcal{I}^{\otimes n}) \le \hat{\gamma}^{\rho}(\Phi_n)$ is proven.

    To prove the lower bound $\gamma^{\rho}(\mathcal{I}^{\otimes n}) \ge \hat{\gamma}^{\rho}(\Phi_n)$, we consider a parallel cut of $n$ wires with NME states, which can be modeled analogously to the cut of one wire with NME states as depicted in \Cref{fig:wire_cut_model}~(c) except that now the Hilbert spaces $A$, $B$, and $C$ are of dimension $2^n$, each consisting of $n$ qubits.
    Consequently, the non-local operator $\mathcal{V}_{ABC}$, which acts on $3n$ qubits, implements $\mathcal{I}^{\otimes n}_{A \rightarrow C}$ when system $A\otimes B$ is traced out afterward.
    This is analogous to the operation defined for the cut of a single wire with NME states in \Cref{eq:non_local_V_with_NME}.
    We aim to demonstrate that a QPD for $\mathcal{V}_{ABC}$, which achieves the optimal sampling overhead $\gamma^{\rho}(\mathcal{I}^{\otimes n})$, can be employed to construct a QPD for the maximally entangled state $\Phi_n$ using the NME state $\rho$, maintaining the same overhead. 
    Therefore, an optimal parallel cut of $n$ wires is given by an optimal QPD for $\mathcal{V}_{ABC}$ , which can be expressed as $\mathcal{V}_{ABC} = \sum_i c_i \mathcal{F}_i$ with $\mathcal{F}_i \in \locc(A \otimes B, C)$.

    \begin{figure}
        \centering
        \includegraphics{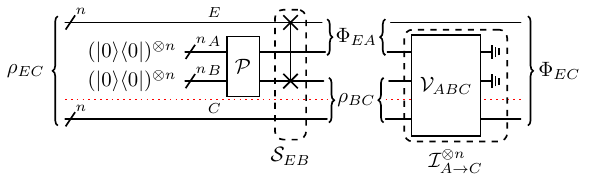}
        \caption{Circuit used in the proof, with qubits separated into subsystems $E \otimes A\otimes B$ and $C$ by the red dotted line.}
        \label{fig:proof_circuit}
    \end{figure}
    
    Consider an additional $n$-qubit Hilbert space $E $, which, together with the Hilbert spaces $A$, $B$, and $C$, forms the composite system $E  \otimes A \otimes B \otimes C$ consisting of $4n$ qubits.
    The system is partitioned into two subsystems: $E  \otimes A \otimes B$ and $C$.
    Initially, within the composite system, $E \otimes C$ hosts an NME state $\rho_{E C}$, while systems $A$ and $B$ are both in the inital state $(\ketbras{0})^{\otimes n}$, as shown on the left side of \Cref{fig:proof_circuit}.
    The bipartite entanglement between the two subsystems is solely characterized by the state $\rho_{E C}$, as it is the only shared state.

    The initial state undergoes a transformation where an operator $\mathcal{P}$ prepares the maximally entangled state $\Phi_{AB}$ between systems $A$ and $B$, starting from the initial state $(\ketbras{0})^{\otimes 2n}_{AB}$.
    Following this, systems $E $ and $B$ exchange states by applying the swap operator $\mathcal{S}_{E B}$.
    This yields the state, illustrated in the middle of \Cref{fig:proof_circuit}, given by
    \begin{align}\label{eq:prepare_and_swap}
        \mathcal{S}_{E B}\left(\mathcal{P}((\ketbras{0})^{\otimes 2n}_{AB}) \otimes \rho_{E C}\right) = \Phi_{E A}\otimes \rho_{BC}.
    \end{align}
    Since $\mathcal{P}$ operates on $A \otimes B$ and $\mathcal{S}_{E B}$ solely changes the states of subsystems $E$ and $B$, their action leaves the system $C$ unaffected. 
    Hence, they are local operators with respect to the defined partition of the subsystems.

    Following the previous state preparation, we execute the identity $\mathcal{I}^{\otimes n}_{A \rightarrow C}$ to produce the maximally entangled state $\Phi_{E C}$ employing $\rho_{BC}$.
    In this process, the state of the system $A$ is transferred to system $C$ through the non-local operator $\mathcal{V}_{ABC}$, as shown on the right side of \Cref{fig:proof_circuit}.
    To quasiprobabilistically simulate $\mathcal{V}_{ABC}$ with local transformations, we employ the QPD, which attains the optimal sampling overhead $\gamma^{\rho}(\mathcal{I}^{\otimes n})$ with $\rho_{BC}$:
    \begin{align}
        \mathcal{V}_{ABC} = \sum_i c_i \mathcal{F}_i,
    \end{align} 
    where $\mathcal{F}_i \in \locc(A\otimes B,C)$.
    This QPD's sampling overhead is expressed as $\sum_i |c_i|$, which is equal to $\gamma^{\rho}(\mathcal{I}^{\otimes n})$, as we assume a QPD with optimal sampling overhead as defined in \Cref{eq:min_sampling_overhead_wire_nme}.

    Therefore, the maximally entangled state $\Phi_{E C}$ between the two subsystems $E  \otimes A\otimes B$ and $C$ is  quasiprobabilistically simulated from the state $\Phi_{E A}\otimes \rho_{BC}$ by leveraging a QPD:
    \begin{align}
        \Phi_{E C} &= (\mathcal{I}^{\otimes n}_{E } \otimes \mathcal{I}^{\otimes n}_{A \rightarrow C})(\Phi_{E A}\otimes \rho_{BC})\\
        &= \operatorname{Tr}_{AB}[(\mathcal{I}^{\otimes n}_{E } \otimes \mathcal{V}_{ABC})(\Phi_{E A}\otimes \rho_{BC})] \\
                &= \sum_i c_i \operatorname{Tr}_{AB}[(\mathcal{I}^{\otimes n}_{E } \otimes \mathcal{F}_i)(\Phi_{E A}\otimes \rho_{BC})].%
    \end{align}
    Thus, the QPD of the maximally entangled state utilizes the local operators $\widetilde{\mathcal{F}}_i  \in \locc(E  \otimes A\otimes B, C)$ defined as
    \begin{align}
        \widetilde{\mathcal{F}}_i((\ketbras{0})^{\otimes 2n}_{AB} \otimes \rho_{E C}) &:= (\mathcal{I}^{\otimes n}_{E } \otimes \mathcal{F}_i)\mathcal{S}_{E B}(\mathcal{P}((\ketbras{0})^{\otimes 2n}_{AB}) \otimes \rho_{E C})\label{eq:local_operator}\\
        &\phantom{:}= (\mathcal{I}^{\otimes n}_{E } \otimes \mathcal{F}_i)(\Phi_{E A}\otimes \rho_{BC}),
    \end{align}
    where \Cref{eq:local_operator} uses \Cref{eq:prepare_and_swap} to prepare the state $\Phi_{E A}\otimes \rho_{BC}$.

    Considering that both the extension of the quantum system by additional qubits in the initial state $(\ketbras{0})^{\otimes 2n}_{AB}$ and the subsequent partial trace operation occur entirely within the same subsystem and are both completely positive and trace-preserving operators, we can define the operators $\mathcal{G}_i(\rho_{E C})$ as follows:
    \begin{align}
        \mathcal{G}_i(\rho_{E C}) =  \operatorname{Tr}_{AB} \left[ \widetilde{\mathcal{F}}_i((\ketbras{0})^{\otimes 2n}_{AB} \otimes \rho_{E C})\right],
    \end{align}
    which inherently belong to $\locc(E ,C)$.
    
    As a consequence, we can establish a QPD for the maximally entangled state between $E $ and $C$ as follows:
    \begin{align}
        \Phi_{E C} &= \sum_i c_i \operatorname{Tr}_{AB}\left[\widetilde{\mathcal{F}}_i((\ketbras{0})^{\otimes 2n}_{AB} \otimes \rho_{E C})\right] \\
        &=  \sum_i c_i \mathcal{G}_i(\rho_{E C}).
    \end{align}
    This construction ensures that the sampling overhead of the given QPD for the maximally entangled state is given by $\sum_i |c_i|$. 
    Therefore, to achieve the minimal sampling overhead, we must have $\hat{\gamma}^{\rho}(\Phi) \leq \sum_i |c_i|$. 
    Since we assumed $\gamma^{\rho}(\mathcal{I}^{\otimes n}) = \sum_i |c_i|$, this implies $\gamma^{\rho}(\mathcal{I}^{\otimes n}) \ge \hat{\gamma}^{\rho}(\Phi)$.
    Consequently, this confirms the lower bound, thereby completing the proof.
\end{proof}

\section{Sampling Overhead for Composite States}\label{sec:appendix_proof_cut_composite_states}
In this section, we demonstrate that the sampling overhead incurred by performing a single parallel wire cut using the composite entangled states is less than the overhead resulting from separate wire cuts for each entangled state. 
To lay the groundwork for our proof, we initially establish a lower bound on the robustness of entanglement for composite states, as detailed in~\Cref{lemma_lower_bound_robustness}.
Following this, we will proceed with the proof of \Cref{theorem_cut_composite_states}.
Subsequently, we will proof \Cref{theorem_composite_max_ent} and \Cref{theorem_composite_separable}.

\begin{lemma}\label{lemma_lower_bound_robustness}
     Let $\{\rho^{(i)}_{A_iB_i}\}_{i=0}^{m-1}$ be a set of $m$ density operators with each $\rho^{(i)}_{A_iB_i}\in D(A_i \otimes B_i)$ for the corresponding $2^{k_i}$-dimensional Hilbert spaces $A_i$ and $B_i$ such that $n= \sum_{i=0}^{m-1}k_i$.
     Then, for the composite state $\bigotimes_{i=0}^{m-1}\rho^{(i)}_{A_iB_i}$, the following lower bound for the robustness of entanglement between the $2^{n}$-dimensional Hilbert spaces $A=\bigotimes_{i=0}^{m-1}A_i$ and $B=\bigotimes_{i=0}^{m-1}B_i$ holds:
    \begin{align}
        R\left(\bigotimes_{i=0}^{m-1}\rho^{(i)}_{A_iB_i}\right) \ge \prod_{i=0}^{m-1}\left(R\left(\rho^{(i)}_{A_iB_i}\right)+1\right) -1.
    \end{align}
\end{lemma}
\begin{proof}
We begin by recalling the fidelity of distillation~\cite{Takagi2024}.
For an arbitrary state $\rho_{AB}\in D(A\otimes B)$, the fidelity of distillation is defined as
\begin{align}
    f_{\locc(A,B)}\left(\rho_{AB}\right) = \max_{\Lambda \in \locc(A,B)}\left\langle\Phi_{AB}\middle|\Lambda\left(\rho_{AB}\right)\middle|\Phi_{AB}\right\rangle
\end{align}
where $\Phi_{AB}$ is the maximally entangled $2n$-qubit state.
Between the fidelity of distillation and the generalized robustness of entanglement, the following relationship is established~\cite[Corollary 9]{Regula2020}:
\begin{align}
    f_{\locc(A,B)}\left(\rho_{AB}\right) = 2^{-n}\left(R\left(\rho_{AB}\right) +1\right) 
\end{align}
or equivalently
\begin{align}
    R\left(\rho_{AB}\right) = 2^{n}f_{\locc(A,B)}\left(\rho_{AB}\right) - 1.
\end{align}

Considering the composite $2n$-qubit state $\bigotimes_{i=0}^{m-1}\rho^{(i)}_{A_iB_i}$, we have:
{\allowdisplaybreaks
    \begin{align}
        R\left(\bigotimes_{i=0}^{m-1}\rho^{(i)}_{A_iB_i}\right) &= 2^{n}f_{\locc(A,B)}\left(\bigotimes_{i=0}^{m-1}\rho^{(i)}_{A_iB_i}\right) -1 \\
        &= 2^{n}\max_{\Lambda \in \locc(A,B)}\left\langle\Phi_{AB}\middle|\Lambda\left(\bigotimes_{i=0}^{m-1}\rho^{(i)}_{A_iB_i}\right)\middle|\Phi_{AB}\right\rangle -1 \\ \label{eq:ineq_composite_space_to_smaller_space}
        &\ge 2^{n}\max_{\substack{\bigotimes_{i=0}^{m-1}\Lambda_i \\ \Lambda_i \in \locc(A_i, B_i)}}\left\langle\Phi_{AB}\middle|\bigotimes_{i=0}^{m-1}\Lambda_i\left(\rho^{(i)}_{A_iB_i}\right)\middle|\Phi_{AB}\right\rangle -1\\ 
        &= 2^{n}\max_{\substack{\bigotimes_{i=0}^{m-1}\Lambda_i, \\ \Lambda_i \in \locc(A_i, B_i)}}\left(\prod_{i=0}^{m-1}\left\langle\Phi_{A_iB_i}\middle|\Lambda_i\left(\rho^{(i)}_{A_iB_i}\right)\middle|\Phi_{A_iB_i}\right\rangle\right) -1 \label{eq:expectation_product_state}\\ 
        &= 2^{n} \prod_{i=0}^{m-1} \left(\max_{\Lambda_i \in \locc(A_i,B_i)}\left\langle\Phi_{A_iB_i}\middle|\Lambda\left(\rho^{(i)}_{A_iB_i}\right)\middle|\Phi_{A_iB_i}\right\rangle\right) -1 \label{eq:exchange_max_and_prod} \\ 
        &= 2^{n} \prod_{i=0}^{m-1}f_{\locc(A_i,B_i)}\left(\rho^{(i)}_{A_iB_i}\right)  -1 \\ 
        &= 2^{n} \prod_{i=0}^{m-1}\left(2^{-k_i}\left(R\left(\rho^{(i)}_{A_iB_i}\right) +1\right)\right)  -1 \\ 
        &= 2^{n}2^{- \left(\sum_{i=0}^{m-1} k_i\right)} \prod_{i=0}^{m-1}\left(R\left(\rho^{(i)}_{A_iB_i}\right) +1\right) -1 \\ 
        &= \prod_{i=0}^{m-1}\left(R\left(\rho^{(i)}_{A_iB_i}\right)+1\right) -1
    \end{align}}
where Inequality~(\ref{eq:ineq_composite_space_to_smaller_space}) holds because $\left\{\bigotimes_{i=0}^{m-1} \Lambda_i \middle| \Lambda_i \in \locc(A_i,B_i)\right\} \subseteq \locc(A,B)$.
Furthermore, the validity of \Cref{eq:expectation_product_state} stems from the fact that $\Phi_{AB} = \bigotimes_{i=0}^{m-1} \Phi_{A_iB_i}$, which enables the fidelity of the composite state $\left\langle\Phi_{AB}\middle|\bigotimes_{i=0}^{m-1}\Lambda_i\left(\rho^{(i)}_{A_iB_i}\right)\middle|\Phi_{AB}\right\rangle$ to be expressed as the product of the fidelities of the individual states $\prod_{i=0}^{m-1}\left\langle\Phi_{A_iB_i}\middle|\Lambda_i\left(\rho^{(i)}_{A_iB_i}\right)\middle|\Phi_{A_iB_i}\right\rangle$.
Lastly, \Cref{eq:exchange_max_and_prod} holds true because the fidelities  $\left\langle\Phi_{A_iB_i}\middle|\Lambda\left(\rho^{(i)}_{A_iB_i}\right)\middle|\Phi_{A_iB_i}\right\rangle$ are non-negative.
This allows the independent maximization of each fidelity, unaffected by the others. 
\end{proof}

In the following, we provide a proof for \Cref{theorem_cut_composite_states}, which is restated for convenience.

\begingroup
\def\thetheorem{\ref{theorem_cut_composite_states}}
\begin{theorem}
    Let $\{\rho^{(i)}_{A_iB_i}\}_{i=0}^{m-1}$ be a set of $m \ge 2$ density operators with each $\rho^{(i)}_{A_iB_i}\in D(A_i \otimes B_i)$ for the corresponding $2^{k_i}$-dimensional Hilbert spaces $A_i$ and $B_i$ such that $n=\sum_{i=0}^{m-1} k_i$.
    Moreover, let each state $\rho^{(i)}_{A_iB_i}$ be an NME state, i.e., $0 \le R\left(\rho^{(i)}_{A_iB_i}\right) < 2^{k_i} -1$.
    Then, the following inequality holds for the composite state $\rho= \bigotimes_{i=0}^{m-1}\rho^{(i)}_{A_iB_i}$:
    \begin{align}
        \prod_{i=0}^{m-1}\gamma^{\rho^{(i)}_{A_iB_i}}(\mathcal{I}^{\otimes k_i}) > \gamma^{\rho}(\mathcal{I}^{\otimes n}). 
    \end{align}
\end{theorem}
\addtocounter{theorem}{-1}
\endgroup

\begin{proof}
    We proceed by induction on $m$. 
    Therefore, we start by examining the base case when $m=2$ with $k_0 + k_1=n$.
    We apply \Cref{theorem_overhead} to the product on the left-hand side:
    {\allowdisplaybreaks
    \begin{align}
       \nonumber&\left(\frac{2^{k_0+1}}{R\left(\rho^{(0)}_{A_0B_0}\right)+1}-1\right)\left(\frac{2^{k_1+1}}{R\left(\rho^{(1)}_{A_1B_1}\right)+1}-1\right)\label{eq:overhead_composite_state_of_two_equality_start}\\
        &=\frac{2^{n+2}}{\left(R\left(\rho^{(0)}_{A_0B_0}\right)+1\right)\left(R\left(\rho^{(1)}_{A_1B_1}\right)+1\right)} - \frac{2^{k_0+1}}{R\left(\rho^{(0)}_{A_0B_0}\right)+1} - \frac{2^{k_1+1}}{R\left(\rho^{(1)}_{A_1B_1}\right)+1} + 1\\
        &=2\frac{2^{n+1}}{\left(R\left(\rho^{(0)}_{A_0B_0}\right)+1\right)\left(R\left(\rho^{(1)}_{A_1B_1}\right)+1\right)} -1 - \frac{2^{k_0+1}}{R\left(\rho^{(0)}_{A_0B_0}\right)+1} - \frac{2^{k_1+1}}{R\left(\rho^{(1)}_{A_1B_1}\right)+1} + 2\\
        \begin{split}
            &=\frac{2^{n+1}}{\left(R\left(\rho^{(0)}_{A_0B_0}\right)+1\right)\left(R\left(\rho^{(1)}_{A_1B_1}\right)+1\right)} -1\\
            &\phantom{=}+\frac{2^{n+1}}{\left(R\left(\rho^{(0)}_{A_0B_0}\right)+1\right)\left(R\left(\rho^{(1)}_{A_1B_1}\right)+1\right)} - \frac{2^{k_0+1}}{R\left(\rho^{(0)}_{A_0B_0}\right)+1} - \frac{2^{k_1+1}}{R\left(\rho^{(1)}_{A_1B_1}\right)+1} + 2
        \end{split}\\
        \begin{split}
            &=\frac{2^{n+1}}{\left(R\left(\rho^{(0)}_{A_0B_0}\right)+1\right)\left(R\left(\rho^{(1)}_{A_1B_1}\right)+1\right)} -1\\
            &\phantom{=}+2\left(\frac{2^{n}}{\left(R\left(\rho^{(0)}_{A_0B_0}\right)+1\right)\left(R\left(\rho^{(1)}_{A_1B_1}\right)+1\right)} - \frac{2^{k_0}}{R\left(\rho^{(0)}_{A_0B_0}\right)+1} - \frac{2^{k_1}}{R\left(\rho^{(1)}_{A_1B_1}\right)+1} + 1\right)
        \end{split}\\
        &=\frac{2^{n+1}}{\left(R\left(\rho^{(0)}_{A_0B_0}\right)+1\right)\left(R\left(\rho^{(1)}_{A_1B_1}\right)+1\right)} -1
          +2\underbrace{\left(\frac{2^{k_0}}{R\left(\rho^{(0)}_{A_0B_0}\right)+1} -1\right)}_{>0}\underbrace{\left(\frac{2^{k_1}}{R\left(\rho^{(1)}_{A_1B_1}\right)+1} - 1\right)}_{>0} \label{eq:overhead_composite_state_of_two_equality}\\
        &>\frac{2^{n+1}}{\left(R\left(\rho^{(0)}_{A_0B_0}\right)+1\right)\left(R\left(\rho^{(1)}_{A_1B_1}\right)+1\right)} -1 \label{eq:greater_than_zero}\\
        &\ge\frac{2^{n+1}}{\left(R\left(\rho^{(0)}_{A_0B_0} \otimes \rho^{(1)}_{A_1B_1}\right)+1\right)} -1 \label{eq:apply_lower_bound_robustness}\\
        &= \gamma^{\rho^{(0)}_{A_0B_0} \otimes \rho^{(1)}_{A_1B_1}}(\mathcal{I}^{\otimes n}).
    \end{align}
    }
    The Inequality~(\ref{eq:greater_than_zero}) holds because we omit a term that is greater than zero, given that $R\left(\rho^{(0)}_{A_0B_0}\right)<2^{k_0}-1$ and $R\left(\rho^{(1)}_{A_1B_1}\right)<2^{k_1}-1$.
    Additionally, \Cref{eq:apply_lower_bound_robustness} is obtained by applying \Cref{lemma_lower_bound_robustness}.

    For the inductive step, assume the theorem holds for some $m \ge 2$ with $\sum_{i=0}^{m-1} k_i = n'$. 
    We show it also holds for $m+1$ states, which are used to cut $n = n' + k_m$ wires. 
    The proof proceeds as follows:
    \begin{align}
        \prod_{i=0}^{m}\gamma^{\rho^{(i)}_{A_iB_i}}(\mathcal{I}^{\otimes k_i})        &=\gamma^{\rho^{(m)}_{A_mB_m}}(\mathcal{I}^{\otimes k_m})\prod_{i=0}^{m-1}\gamma^{\rho^{(i)}_{A_iB_i}}(\mathcal{I}^{\otimes k_i})\\
        &>\gamma^{\rho^{(m)}_{A_mB_m}}(\mathcal{I}^{\otimes k_i})\gamma^{\bigotimes_{i=0}^{m-1}\rho^{(i)}_{A_iB_i}}(\mathcal{I}^{\otimes n'})
    \end{align}
    where the induction hypothesis has been applied in the last step.
    Following this, we rewrite the expression using \Cref{theorem_overhead}:
    {\allowdisplaybreaks
    \begin{align}
        \nonumber&\prod_{i=0}^{m}\gamma^{\rho^{(i)}_{A_iB_i}}(\mathcal{I}^{\otimes k_i})\\
        &> \left(\frac{2^{k_m+1}}{R\left(\rho^{(m)}_{A_mB_m}\right)+1}-1\right)\left(\frac{2^{n'+1}}{R\left(\bigotimes_{i=0}^{m-1}\rho^{(i)}_{A_iB_i}\right)+1}-1\right)\\
        &=\frac{2^{n+2}}{\left(R\left(\rho^{(m)}_{A_mB_m}\right)+1\right)\left(R\left(\bigotimes_{i=0}^{m-1}\rho^{(i)}_{A_iB_i}\right)+1\right)}-\frac{2^{k_m+1}}{R\left(\rho^{(m)}_{A_mB_m}\right)+1} - \frac{2^{n'+1}}{R\left(\bigotimes_{i=0}^{m-1}\rho^{(i)}_{A_iB_i}\right)+1} +1\\
        \begin{split}
            &=\frac{2^{n+1}}{\left(R\left(\rho^{(m)}_{A_mB_m}\right)+1\right)\left(R\left(\bigotimes_{i=0}^{m-1}\rho^{(i)}_{A_iB_i}\right)+1\right)}-1 \\
            &\phantom{=}+\frac{2^{n+1}}{\left(R\left(\rho^{(m)}_{A_mB_m}\right)+1\right)\left(R\left(\bigotimes_{i=0}^{m-1}\rho^{(i)}_{A_iB_i}\right)+1\right)}-\frac{2^{k_m+1}}{R\left(\rho^{(m)}_{A_mB_m}\right)+1} - \frac{2^{n'+1}}{R\left(\bigotimes_{i=0}^{m-1}\rho^{(i)}_{A_iB_i}\right)+1} +2
        \end{split}\\ 
        \begin{split}
            &=\frac{2^{n+1}}{\left(R\left(\rho^{(m)}_{A_mB_m}\right)+1\right)\left(R\left(\bigotimes_{i=0}^{m-1}\rho^{(i)}_{A_iB_i}\right)+1\right)}-1 \\
            &\phantom{=}+2\left(\frac{2^{n}}{\left(R\left(\rho^{(m)}_{A_mB_m}\right)+1\right)\left(R\left(\bigotimes_{i=0}^{m-1}\rho^{(i)}_{A_iB_i}\right)+1\right)}-\frac{2^{k_m}}{R\left(\rho^{(m)}_{A_mB_m}\right)+1} - \frac{2^{n'}}{R\left(\bigotimes_{i=0}^{m-1}\rho^{(i)}_{A_iB_i}\right)+1} +1\right)
        \end{split}\\
        \begin{split}
            &=\frac{2^{n+1}}{\left(R\left(\rho^{(m)}_{A_mB_m}\right)+1\right)\left(R\left(\bigotimes_{i=0}^{m-1}\rho^{(i)}_{A_iB_i}\right)+1\right)}-1\\
            &\phantom{=}+2\underbrace{\left(\frac{2^{k_m}}{R\left(\rho^{(m)}_{A_mB_m}\right)+1} -1\right)}_{>0}\underbrace{\left(\frac{2^{n'}}{R\left(\bigotimes_{i=0}^{m-1}\rho^{(i)}_{A_iB_i}\right)+1} -1\right)}_{>0}.  \label{eq:more_than_zero_for_nme}
        \end{split}
    \end{align}}
    The last summand in \Cref{eq:more_than_zero_for_nme} is strictly positive for NME states $\rho^{(m)}_{A_mB_m}$ and $\bigotimes_{i=0}^{m-1}\rho^{(i)}_{A_iB_i}$.
    Thus, we can omit this summand and obtain the following inequality:
    \begin{align}
        \prod_{i=0}^{m}\gamma^{\rho^{(i)}_{A_iB_i}}(\mathcal{I}^{\otimes k_i}) &>\frac{2^{n+1}}{\left(R\left(\rho^{(m)}_{A_mB_m}\right)+1\right)\left(R\left(\bigotimes_{i=0}^{m-1}\rho^{(i)}_{A_iB_i}\right)+1\right)}-1\\
        &\ge\left(\frac{2^{n+1}}{R\left(\bigotimes_{i=0}^{m}\rho^{(i)}_{A_iB_i}\right)+1}-1\right) \label{eq:apply_lower_bound_robustness_2}\\
        &= \gamma^{\rho}(\mathcal{I}^{\otimes n}). \label{eq:apply_theorem_overhead}
    \end{align}
    In this derivation, \Cref{lemma_lower_bound_robustness} is applied to achieve \Cref{eq:apply_lower_bound_robustness_2}, and \Cref{eq:apply_theorem_overhead} follows from the application of \Cref{theorem_overhead}.
    This completes the inductive step, proving the theorem for all $m \ge 2$. 
\end{proof}

To prove \Cref{theorem_composite_max_ent}, the following lemma shows that the lower bound established in \Cref{lemma_lower_bound_robustness} for the robustness of entanglement in composite states is exact when one component of the composite state is maximally entangled.
\begin{lemma}\label{lemma_robustness_composite_max_ent}
    Consider a state $\rho_{A_0B_0}\in D(A_0\otimes B_0)$ and a maximally entangled state $\Phi_{A_1B_1}\in D(A_1, B_1)$ for Hilbert spaces $A_0,B_0,A_1$, and $B_1$.
    Then for the generalized robustness of entanglement, it holds that
    \begin{align}
        R(\rho_{A_0B_0} \otimes \Phi_{A_1B_1}) = \left(R(\rho_{A_0B_0}) + 1\right)\left(R( \Phi_{A_1B_1})+1\right) - 1.
    \end{align}
\end{lemma}
\begin{proof}
    Using the general robustness of entanglement, Datta~\cite{Datta2008} introduced an entanglement monotone known as the log robustness, defined by
    \begin{align}\label{eq:log_robustness}
        LR(\rho) := \log(1 + R(\rho)).
    \end{align}
    This entanglement monotone is generally not additive for arbitrary composite states $\rho_{A_0B_0} \otimes \rho_{A_1B_1}$, where $\rho_{A_0B_0}\in D(A_0\otimes B_0)$ and $\rho_{A_1B_1}\in D(A_1\otimes B_1)$, meaning that $LR(\rho_{A_0B_0} \otimes \rho_{A_1B_1}) \neq LR(\rho_{A_0B_0}) + LR(\rho_{A_1B_1})$.
    However, additivity does hold under specific conditions, notably when one of the states is a bipartite pure state $\ket{\Psi^\alpha}$~\cite[Theorem 9]{Rubboli2022}.
    Since the maximally entangled state $\Phi_{A_1B_1}\in D(A_1\otimes B_1)$ is a bipartite pure state, it follows for an arbitrary state $\rho_{A_0B_0}\in D(A_0\otimes B_0)$ that
    \begin{align}
        LR(\rho_{A_0B_0} \otimes \Phi_{A_1B_1}) = LR(\rho_{A_0B_0}) + LR(\Phi_{A_1B_1}).
    \end{align}
    By applying the definition of the log robustness from \Cref{eq:log_robustness}, we derive
    \begin{align}
        \log(1 + R(\rho_{A_0B_0} \otimes \Phi_{A_1B_1})) &= \log(1 + R(\rho_{A_0B_0})) + \log(1 + R(\Phi_{A_1B_1}))\\
        &= \log\left((1 + R(\rho_{A_0B_0}))(1 + R(\Phi_{A_1B_1}))\right).
    \end{align}
    This equation implies
    \begin{align}
        1 + R(\rho_{A_0B_0} \otimes \Phi_{A_1B_1}) &= (1 + R(\rho_{A_0B_0}))(1 + R(\Phi_{A_1B_1})).
    \end{align}
    Subtracting 1 from both sides yields the desired result.
\end{proof}

The following proof of \Cref{theorem_composite_max_ent} leverages the result of \Cref{lemma_robustness_composite_max_ent} to establish the optimal sampling overhead when using a composite state consisting of an arbitrary state and a maximally entangled state for cutting multiple wires in a quantum circuit.
The theorem is restated for convenience.
\begingroup
\def\thetheorem{\ref{theorem_composite_max_ent}}
\begin{theorem}
    Consider an arbitrary bipartite state $\rho_{A_0B_0}\in D(A_0\otimes B_0)$ and a maximally entangled state $\Phi_{A_1B_1}\in D(A_1 \otimes B_1)$, where $A_0$ and $B_0$ are $2^{n_0}$-dimensional Hilbert spaces and $A_1$ and $B_1$ are $2^{n_1}$-dimensional Hilbert spaces. 
    For the sampling overhead of cutting $n$ wires, where $n = n_0 + n_1$, it holds that
    \begin{align}
        \gamma^{\rho_{A_0B_0} \otimes \Phi_{A_1B_1}}(\mathcal{I}^{\otimes n}) 
        &= \gamma^{\rho_{A_0B_0}}(\mathcal{I}^{\otimes n_0})\gamma^{\Phi_{A_1B_1}}(\mathcal{I}^{\otimes n_1}) \\
        &= \gamma^{\rho_{A_0B_0}}(\mathcal{I}^{\otimes n_0})
    \end{align}
\end{theorem}
\addtocounter{theorem}{-1}
\endgroup
\begin{proof}
We start by applying \Cref{theorem_overhead} and then use \Cref{lemma_robustness_composite_max_ent}:
    \begin{align}
        \gamma^{\rho_{A_0B_0} \otimes \Phi_{A_1B_1}}(\mathcal{I}^{\otimes n}) &= \frac{2^{n+1}}{R(\rho_{A_0B_0} \otimes \Phi_{A_1B_1}) + 1} - 1\\
        &= \frac{2^{n+1}}{(1 + R(\rho_{A_0B_0}))(1 + R(\Phi_{A_1B_1}))} - 1\\
        &= \frac{2^{n+1}}{(1 + R(\rho_{A_0B_0}))2^{n_1}} - 1 \label{eq:robustness_max_ent}\\
        &= \frac{2^{n_0+1}}{1 + R(\rho_{A_0B_0})} - 1\\
        &= \gamma^{\rho_{A_0B_0}}(\mathcal{I}^{\otimes n_0}).
    \end{align}
\Cref{eq:robustness_max_ent} uses $R(\Phi_{A_1B_1}) = 2^{n_1} -1$. 
The final equality is obtained by reapplying \Cref{theorem_overhead}.
Moreover, since $\gamma^{\Phi_{A_1B_1}}(\mathcal{I}^{\otimes n_1}) = 1$, it holds that
\begin{align}
    \gamma^{\rho_{A_0B_0}}(\mathcal{I}^{\otimes n_0}) = \gamma^{\rho_{A_0B_0}}(\mathcal{I}^{\otimes n_0})\gamma^{\Phi_{A_1B_1}}(\mathcal{I}^{\otimes n_1}).
\end{align}
\end{proof}

Before proceeding with the proof of \Cref{theorem_composite_separable}, we prove the following  lemma.
\begin{lemma}\label{lemma_robustness_composite_sep}
    Consider an entangled state $\rho_{A_eB_e}\in D(A_e\otimes B_e)$ and a separable state $\tau_{A_sB_s}\in S(A_s, B_s)$ for Hilbert spaces $A_e,B_e,A_s$, and $B_s$.
    Then for the generalized robustness of entanglement, it holds that
    \begin{align}
        R(\rho_{A_eB_e}) = R(\rho_{A_eB_e} \otimes \tau_{A_sB_s}).
    \end{align}
\end{lemma}
\begin{proof}
    By using \Cref{lemma_lower_bound_robustness}, we obtain the following lower bound:
    \begin{align}
        R(\rho_{A_eB_e} \otimes \tau_{A_sB_s}) &\ge (R(\rho_{A_eB_e}) +1)(R(\tau_{A_sB_s}) +1) -1 \\
        &= R(\rho_{A_eB_e})\label{eq:apply_rob_sep_state},
    \end{align}
    where we arrived at \Cref{eq:apply_rob_sep_state} by using that $\tau_{A_sB_s}$ is a separable state, and thus has no entanglement, i.e., $R(\tau_{A_sB_s})=0$.
    
    To show that $R(\rho_{A_eB_e} \otimes \tau_{A_sB_s}) \le R(\rho_{A_eB_e})$ is also an upper bound, consider the definition of the generalized robustness of $\rho_{A_eB_e}$ in \Cref{eq:robustness}.
    According to this definition, there exists a state $\omega_{A_eB_e}\in D(A_e \otimes B_e)$ such that, when combined linearly with $\rho_{A_eB_e}$ using $\lambda = R(\rho_{A_eB_e})$, the resulting state is a separable state.
    This separable state, referred to as $\varrho_{A_eB_e}\in S(A_e, B_e)$ in the following, is given by:
    \begin{align}
        \varrho_{A_eB_e} = \frac{\rho_{A_eB_e} + R(\rho_{A_eB_e}) \omega_{A_eB_e} }{1+R(\rho_{A_eB_e})}.
    \end{align}
    Thus, it holds for the separable state $\varrho_{A_eB_e}\otimes \tau_{A_sB_s}$ that
    \begin{align}
        \varrho_{A_eB_e}\otimes \tau_{A_sB_s} &= \left(\frac{\rho_{A_eB_e} + R(\rho_{A_eB_e}) \omega_{A_eB_e} }{1+R(\rho_{A_eB_e})}\right)\otimes \tau_{A_sB_s}\\
        &= \frac{\rho_{A_eB_e}\otimes \tau_{A_sB_s} + R(\rho_{A_eB_e})  \omega_{A_eB_e} \otimes \tau_{A_sB_s} }{1+R(\rho_{A_eB_e})}\label{eq:proof_robustness_sep}.
    \end{align}
    \Cref{eq:proof_robustness_sep} is a mixture as required in the definition of the generalized robustness of entanglement in \Cref{eq:robustness} with $\lambda = R(\rho_{A_eB_e})$.
    However, it is not guaranteed to be the optimal mixture, thus $R(\rho_{A_eB_e} \otimes \tau_{A_sB_s}) \le R(\rho_{A_eB_e})$.
    Together with the lower bound shown above, this implies $R(\rho_{A_eB_e} \otimes \tau_{A_sB_s}) = R(\rho_{A_eB_e})$.
\end{proof}

We proceed with the proof of \Cref{theorem_composite_separable}, which is restated for convenience in the following.
\begingroup
\def\thetheorem{\ref{theorem_composite_separable}}
\begin{theorem}
    Consider an entangled state $\rho_{A_eB_e}\in D(A_e\otimes B_e)$, where $A_e$ and $B_e$ are $2^{n_e}$-dimensional Hilbert spaces. 
    To cut $n$ wires, with $n>n_e$, consider the composite state $\rho_{A_eB_e} \otimes \tau_{A_sB_s}$, which augments the entangled state $\rho_{A_eB_e}$ with a separable state $\tau_{A_sB_s} \in S(A_s,B_s)$. 
    Here, $A_s$ and $B_s$ are $2^{n_s}$-dimensional Hilbert spaces, where $n_s= n-n_e$.
    The sampling overhead of a single parallel cut using the composite state $\rho_{A_eB_e} \otimes \tau_{A_sB_s}$ is lower than the overhead resulting from a parallel cut of $n_e$ wires with $\rho_{A_eB_e}$ and separately cutting the remaining $n_s$ wires without entanglement.
    The advantage in the sampling overhead is given by
    \begin{align}
        \gamma^{\rho_{A_eB_e}}(\mathcal{I}^{\otimes n_{e}})\gamma(\mathcal{I}^{\otimes n_{s}}) - \gamma^{\rho_{A_eB_e}\otimes \tau_{A_sB_s}}(\mathcal{I}^{\otimes n})
        &= \left(\gamma^{\rho_{A_eB_e}}(\mathcal{I}^{\otimes n_{e}}) - 1\right)(2^{n_s} - 1) \ge 0.
    \end{align}
\end{theorem}
\addtocounter{theorem}{-1}
\endgroup
\begin{proof}
    Similar to Equations~(\ref{eq:overhead_composite_state_of_two_equality_start}) to (\ref{eq:overhead_composite_state_of_two_equality}) in the proof of \Cref{theorem_cut_composite_states}, we obtain the following equality:
    \begin{align}
        \gamma^{\rho_{A_eB_e}}(\mathcal{I}^{\otimes n_{e}})\gamma(\mathcal{I}^{\otimes n_{s}})
        &=\left(\frac{2^{n_{e}+1}}{R\left(\rho_{A_eB_e}\right)+1}-1\right)\left(2^{n_{s}+1}-1\right)\\
        &= \frac{2^{n+2}}{R\left(\rho_{A_eB_e}\right)+1} - \frac{2^{n_{e}+1}}{R\left(\rho_{A_eB_e}\right)+1} - 2^{n_{s}+1} + 1\\
        &= \frac{2^{n+1}}{R\left(\rho_{A_eB_e}\right)+1} -1 + 2\left(\frac{2^{n}}{R\left(\rho_{A_eB_e}\right)+1} - \frac{2^{n_{e}}}{R\left(\rho_{A_eB_e}\right)+1} - 2^{n_{s}} + 1\right)\\
        &=\frac{2^{n+1}}{R\left(\rho_{A_eB_e}\right)+1} -1 + 2\left(\frac{2^{n_{e}}}{R\left(\rho_{A_eB_e}\right)+1} -1\right)\left(2^{n_{s}} - 1\right).
      \end{align}
    Given that $\tau_{A_sB_s}$ is a separable state, we can apply \Cref{lemma_robustness_composite_sep} with $R\left(\rho_{A_eB_e}\right)= R\left(\rho_{A_eB_e}\otimes \tau_{A_sB_s}\right)$, leading to
    \begin{align}
        \gamma^{\rho_{A_eB_e}}(\mathcal{I}^{\otimes n_{e}})\gamma(\mathcal{I}^{\otimes n_{s}}) &=\frac{2^{n+1}}{R\left(\rho_{A_eB_e}\otimes \tau_{A_sB_s}\right)+1} -1 + 2\left(\frac{2^{n_{e}}}{R\left(\rho_{A_eB_e}\right)+1} -1\right)\left(2^{n_{s}} - 1\right)\\
        &=\gamma^{\rho_{A_eB_e} \otimes \tau_{A_sB_s}}(\mathcal{I}^{\otimes n}) + \left(\frac{2^{n_{e}+1}}{R\left(\rho_{A_eB_e}\right)+1} -2\right)\left(2^{n_{s}} - 1\right)\label{eq:def_sampling_overhead}\\
        &=\gamma^{\rho_{A_eB_e} \otimes \tau_{A_sB_s}}(\mathcal{I}^{\otimes n})+\left(\gamma^{\rho_{A_eB_e}}(\mathcal{I}^{\otimes n_{e}})-1\right)\left(2^{n_{s}} - 1\right)\label{eq:def_sampling_overhead_2}
    \end{align}
    where \Cref{eq:def_sampling_overhead,eq:def_sampling_overhead_2} are obtained by using the results on the sampling overhead from \Cref{theorem_overhead}. %
    By rearranging the equation above, we obtain the formulation of \Cref{theorem_composite_separable}.
    Additionally, $\left(\gamma^{\rho_{A_eB_e}}(\mathcal{I}^{\otimes n_{e}})-1\right)\left(2^{n_{s}} - 1\right)$ is non-negative, as both factors are non-negative. 
\end{proof}

\section{Proof of \Cref{theorem_decomposition}}\label{sec:proof_decomposition}
In this section, we first introduce three foundational lemmas that serve as the underpinnings for the subsequent proof of \Cref{theorem_decomposition}.

\begin{lemma}\label{lemma_1}
Let $j,k \in \mathbb{GF}(2^n)$ and consider the unitary operators $U_j$, $S_{j,k}$, and $\hat{Z}_k$, as defined in \Cref{eq:phase_op,eq:S_jk,eq:mub_basis_transformation}, respectively.
Then, the following relation holds:
    \begin{align}
        U_j \hat{Z}_k U_j^{\dagger} = S_{j,k}.
    \end{align}
\end{lemma}

\begin{proof}
\begin{align}
U_j \hat{Z}_k U_j^{\dagger} %
&=\left(\sum_{l=0}^{2^n-1}\ketbra{e_l^j}{l}\right)\left(\sum_{u = 0}^{2^n -1} (-1)^{u\odot k}\ketbras{u}\right)\left(\sum_{l'=0}^{2^n-1}\ketbra{l'}{e_{l'}^j}\right) \\
&=\sum_{l,u,l'=0}^{2^n-1}(-1)^{u\odot k}\ket{e_l^j}\!\braket{l|u}\!\braket{u|l'}\!\bra{e_{l'}^j} \\
&= \sum_{l=0}^{2^n-1} (-1)^{l \odot k}\ketbras{e_{l}^j}\\
&= S_{j,k}
\end{align}
\end{proof}

Having established the relationship between unitary operators $U_j$, $\hat{Z}_k$, and $S_{j,k}$ in \Cref{lemma_1}, we now turn our attention to the phase factors $s_{j,k}$ in \Cref{lemma_s_jk_squared}.

\begin{lemma}\label{lemma_s_jk_squared}
Let $s_{j,k} = \prod_{0 \le r,t \le n-1 } \iu^{j\odot(k_r2^r)\odot (k_t2^t)}$  be the phase factor as introduced \Cref{eq:phase_factor}. 
Then, it holds that
    \begin{align}
        (s_{j,k})^2 = (-1)^{j \odot k \odot k}.
    \end{align}
\end{lemma}
\begin{proof}
    \begin{align}
        (s_{j,k})^2 &= \left(\prod_{0 \le r,t \le n-1 } \iu^{j\odot(k_r2^r)\odot (k_t2^t)}\right)^2\\
        &= \prod_{0 \le r,t \le n-1 } \left(\iu^{j\odot(k_r2^r)\odot (k_t2^t)}\right)^2\\
        &= \prod_{0 \le r,t \le n-1 } \left(\iu^2\right)^{j\odot(k_r2^r)\odot (k_t2^t)}\\
        &= \prod_{0 \le r,t \le n-1 } \left(-1\right)^{j\odot(k_r2^r)\odot (k_t2^t)}\\
        &= \left(-1\right)^{\bigoplus_{0 \le r,t \le n-1} j\odot(k_r2^r)\odot (k_t2^t)}\label{eq:sum_in_exponent}\\
        &= \left(-1\right)^{j\odot \left(\bigoplus_{r=0}^{n-1}k_r2^r\right)\odot \left(\bigoplus_{t=0}^{n-1} k_t2^t\right)}\label{eq:gf_distributivity}\\
        &= \left(-1\right)^{j\odot k \odot k}
    \end{align}
    where \Cref{eq:sum_in_exponent} follows from \Cref{eq:gf_exponent_sum} and \Cref{eq:gf_distributivity} is obtained by using the distributivity of multiplication and addition in the Galois field.
\end{proof}

Central to the forthcoming proof, \Cref{lemma_sum_S_im} examines the impact of applying the $2^n$ distinct unitary operators $\{S_{j,k}\}_{j=0}^{2^n-1}$, i.e., the $k$-th operators from all commuting sets, to a density operator. 
This lemma demonstrates that the effect of applying the operators from the set uniformly at random, i.e., with probability $2^{-n}$, is equal to the action of a measurement followed by a state preparation.
\begin{lemma}\label{lemma_sum_S_im}
Let $\rho\in D(A)$ for a Hilbert space $A$ with $\dim(A)= 2^n$ and $k\in \mathbb{GF}(2^n)$ with $k\ne 0$.
Moreover, $\{S_{j,k}\}_{j=0}^{2^n-1}$ are the unitary operators as defined in \Cref{eq:S_jk}.
Then, it holds that
    \begin{align}
        \frac{1}{2^n}\sum_{j=0}^{2^n-1}S_{j,k} \rho S_{j,k} = \sum_{l=0}^{2^n-1}\tr[\ketbras{l}\rho]\ketbras{l\oplus k}.
    \end{align}
\end{lemma}

\begin{proof}
    We start by using the definition of the operators $S_{j,k}$ as given in \Cref{eq:S_jk}, which results in
    \begin{align}
    \sum_{j=0}^{2^n-1}S_{j,k} \rho S_{j,k} =& \sum_{j=0}^{2^n-1} (s_{j,k})^2 \hat{Z}_{j\odot k}\hat{X}_k \rho \hat{Z}_{j\odot k}\hat{X}_k \\
    =& \sum_{j=0}^{2^n-1} (-1)^{j \odot k \odot k} \hat{Z}_{j\odot k}\hat{X}_k \rho \hat{Z}_{j\odot k}\hat{X}_k 
    \end{align}
    where we applied \Cref{lemma_s_jk_squared}.
    We proceed by applying the commutation relation of \Cref{lemma_commutation_XZ} in the following:
    \begin{align}
    \sum_{j=0}^{2^n-1}S_{j,k} \rho S_{j,k} &= \sum_{j=0}^{2^n-1} (-1)^{j \odot k \odot k} (-1)^{j \odot k \odot k} \hat{X}_k \hat{Z}_{j\odot k} \rho \hat{Z}_{j\odot k}\hat{X}_k \\
    &= \sum_{j=0}^{2^n-1}  \hat{X}_k \hat{Z}_{j\odot k} \rho \hat{Z}_{j\odot k}\hat{X}_k \\
    &= \hat{X}_k \left(\sum_{j=0}^{2^n-1}  \hat{Z}_{j\odot k} \rho \hat{Z}_{j\odot k} \right)\hat{X}_k. \label{eq:xzrhozx}
    \end{align}
    Focusing on the summation within the brackets, we derive the following expression:
    \begin{align}
        \sum_{j=0}^{2^n-1}  \hat{Z}_{j\odot k} \rho \hat{Z}_{j\odot k} &= \sum_{j=0}^{2^n-1}\left(\sum_{l=0}^{2^n-1}(-1)^{l\odot j\odot k}\ketbras{l}\right)\rho \left(\sum_{l'=0}^{2^n-1}(-1)^{l'\odot j\odot k}\ketbras{l'}\right) \\
        &= \sum_{l,l'=0}^{2^n-1}\sum_{j=0}^{2^n-1}(-1)^{(l\oplus l')\odot j\odot k}\ketbras{l}\rho\ketbras{l'} \\
        &= \sum_{l,l'=0}^{2^n-1}2^n\delta_{(l\oplus l')\odot k, 0}\ketbras{l}\rho\ketbras{l'} \label{eq:apply_galois_sum}\\
        &= \sum_{l=0}^{2^n-1}2^n\ketbras{l}\rho\ketbras{l} \label{eq:l_equal_l'}
    \end{align}
    where \Cref{eq:apply_galois_sum} results from the application of \Cref{lemma_sum_galois}.
    Subsequently, \Cref{eq:l_equal_l'} follows from the fact that for $k\ne 0$, it holds that  \mbox{$(l\oplus l')\odot k = 0\Leftrightarrow l\oplus l' = 0 \Leftrightarrow l=l'$} since in $\mathbb{GF}({2^n})$ each element is its own additive inverse.
    Substituting \Cref{eq:l_equal_l'} in \Cref{eq:xzrhozx}, results in
    \begin{align}
        \sum_{j=0}^{2^n-1}S_{j,k} \rho S_{j,k} &= \hat{X}_k \left(\sum_{l=0}^{2^n-1}2^n\ketbras{l}\rho\ketbras{l}\right)\hat{X}_k \\
        &= \left(\sum_{s=0}^{2^n-1}\ket{s}\bra{s\oplus k}\right) \left(\sum_{l=0}^{2^n-1}2^n\ketbras{l}\rho\ketbras{l}\right)\left(\sum_{t=0}^{2^n-1}\ket{t}\bra{t\oplus k}\right)\\
        &= \left(\sum_{s'=0}^{2^n-1}\ket{s' \oplus k}\bra{s'}\right) \left(\sum_{l=0}^{2^n-1}2^n\ketbras{l}\rho\ketbras{l}\right)\left(\sum_{t=0}^{2^n-1}\ket{t}\bra{t\oplus k}\right) \label{eq:j_prime}\\
        &= 2^n\sum_{l=0}^{2^n-1}\ketbra{l\oplus k}{l}\rho\ketbra{l}{l\oplus k}\\
        &= 2^n\sum_{l=0}^{2^n-1}\braket{l|\rho|l}\ketbras{l\oplus k}\\
        &= 2^n\sum_{l=0}^{2^n-1}\tr[\ketbras{l}\rho]\ketbras{l\oplus k}
    \end{align}
    where \Cref{eq:j_prime} is derived by substituting $s' = s \oplus k$. 
    This substitution is permissible since the sum over the substituted terms again constitutes a sum over all elements of the field $\mathbb{GF}(2^n)$.
    Moreover, it holds that $s' \oplus k = s \oplus k \oplus k = s$.
\end{proof}

In the following, we proceed with the proof of \Cref{theorem_decomposition}, which is restated for convenience.
\begingroup
\def\thetheorem{\ref{theorem_decomposition}}
\begin{theorem}
    Let $\ket{\Psi^{\alpha}}$ denote a pure NME state of $2n$ qubits with entanglement robustness $R(\Psi^{\alpha})< 2^n -1$. 
The $n$-qubit identity operator $\mathcal{I}^{\otimes n}$ can be decomposed using quantum teleportation described by $\mathcal{E}_{\text{tel}}^{\Psi^{\alpha}}$ with $\Psi^{\alpha}$ as resource state as follows:
\begin{equation}
    \mathcal{I}^{\otimes n}(\bullet) = \frac{1}{R(\Psi^{\alpha})+1}\sum_{j=0}^{2^n-1} U_j\mathcal{E}_{\text{tel}}^{\Psi^{\alpha}}\left(U_j^{\dagger}(\bullet) U_j\right)U_j^{\dagger} - \left(\frac{2^n}{R(\Psi^{\alpha})+1}-1\right)\sum_{l=0}^{2^n-1} \tr\left[ \ketbras{l}(\bullet) \right]\rho_l,
\end{equation}
where $U_j$ are the unitary operators from \Cref{eq:mub_basis_transformation} that transform the computational basis in the different MUBs and the density operator $\rho_l$ is 
\begin{align}
    \rho_l := \sum_{k=1}^{2^n-1}\Pr(k|\alpha)\ketbras{l \oplus k}
\end{align}
with the probabilities $\Pr(k|\alpha)$ for selecting $0<k<2^n$, calculated from the $2^n$-dimensional Schmidt vector $\alpha$ as 
\begin{align}
    \Pr(k|\alpha):=\frac{\left(\sum_{j=0}^{2^n-1}(-1)^{k \odot j}\alpha_j\right)^2}{2^n-1 -R(\Psi^{\alpha})}.
\end{align}
This decomposition achieves the optimal sampling overhead from \Cref{theorem_overhead}.
\end{theorem}
\addtocounter{theorem}{-1}
\endgroup

\begin{proof}
As introduced in \Cref{eq:tele2}, the result of the teleportation of an $n$-qubit state with a pure $2n$-qubit NME state $\ket{\Psi^{\alpha}}$ depends on the overlap with the generalized Bell basis states $\ket{\Phi^\sigma}$ with $\sigma = \bigotimes_{i=0}^{n-1}\sigma_i$ and $\sigma_i \in \{I, X, Y, Z\}$.
By applying the definitions of $\ket{\Phi_n}$, $\ket{\Psi^{\alpha}}$, and $\ket{\Phi^{\sigma}}$ from \Cref{eq:max_ent_state,eq:phi_sigma,eq:nme}, respectively, we obtain
\begin{align}
    \braket{\Phi^\sigma|\Psi^{\alpha}|\Phi^\sigma} &= |\!\braket{\Phi^\sigma|\Psi^{\alpha}}\!|^2 \\
    &= |\!\braket{\Phi_n|(\sigma \otimes I^{\otimes n})|\Psi^{\alpha}}\!|^2 \label{eq:def_Phi_sigma}\\
    &= \left|\sum_{j,k=0}^{2^n-1} \frac{1}{\sqrt{2^n}}\bra{j}\otimes\bra{j}(\sigma \otimes I^{\otimes n})\;\alpha_j \ket{k}\otimes\ket{k}\right|^2 \label{eq:def_Phi_n_alpha}\\
    &= \left|\sum_{j,k=0}^{2^n-1} \frac{\alpha_j}{\sqrt{2^n}}\braket{j|\sigma|k} \braket{j|k}\right|^2 \label{eq:tensorproduct_factorization}\\
    &= \left|\sum_{j=0}^{2^n-1}\frac{\alpha_j}{\sqrt{2^n}} \braket{j|\sigma|j}\right|^2 \\
    &= \frac{1}{2^n}\left|\sum_{j=0}^{2^n-1}\alpha_j \braket{j|\sigma|j}\right|^2 \\
    &= \frac{1}{2^n}\left(\sum_{j=0}^{2^n-1}\alpha_j \braket{j|\sigma|j}\right)^2 \label{eq:real_sum}.
\end{align}
In \Cref{eq:tensorproduct_factorization}, we utilize that $\ket{j}\otimes \ket{j}$ and $(\sigma \otimes I^{\otimes n})\ket{k}\otimes\ket{k}$ are product states. 
The inner product of product states can be factorized, resulting in $\bra{j}\otimes\bra{j}(\sigma \otimes I^{\otimes n})\ket{k}\otimes\ket{k} = \braket{j|\sigma|k} \braket{j|I^{\otimes n}|k}$. 
Additionally, \Cref{eq:real_sum} is valid because the sum $\sum_{j=0}^{2^n-1}\alpha_j \braket{j|\sigma|j}$ results in a real number since the Schmidt coefficients $\alpha_j$ and the diagonal elements of the Hermitian matrix $\sigma$, denoted by $\braket{j|\sigma|j}$, are both real.

We evaluate the inner sum of \Cref{eq:real_sum} further by considering the binary representation of $j$, expressed as $j=\sum_{i=0}^{n-1} j_i 2^i$ with $j_i \in \{0,1\}$. 
This leads to the following product:
\begin{align}\label{eq:sigma_product}
    \braket{j|\sigma|j} &= \prod_{i=0}^{n-1}\braket{j_i|\sigma_i|j_i}.
\end{align}
For operators $\sigma_i \in \{X, Y\}$, we observe that $\braket{j_i|\sigma_i|j_i}=0$, as all diagonal elements of Pauli $X$ and $Y$ matrices are zero.
This implies with \Cref{eq:sigma_product} that for any operator $\sigma \in \{I, X, Y,Z\}^{\otimes n} \setminus  \{I, Z\}^{\otimes n}$,  the expression $\braket{i|\sigma|i}$ is zero. 
Therefore, the overlap in this case is 
\begin{align}\label{eq:XY_zero}
    \braket{\Phi^\sigma|\Psi^{\alpha}|\Phi^\sigma}= 0 \quad\text{for }\sigma \in \{I, X, Y,Z\}^{\otimes n} \setminus  \{I, Z\}^{\otimes n}.
\end{align}

Conversely, for $\sigma \in \{I, Z\}^{\otimes n}$, there exists a $k \in \mathbb{GF}(2^n)$ for the phase operator $\hat{Z}_k$ such that $\hat{Z}_k = \sigma$ as shown in \Cref{eq:phase_op_pauli_prod}.
Thus, by substituting the definition of $\hat{Z}_k$ in \Cref{eq:real_sum}, we obtain
\begin{align}
    \braket{\Phi^{\hat{Z}_k}|\Psi^{\alpha}|\Phi^{\hat{Z}_k}} &= \frac{1}{2^n}\left(\sum_{i=0}^{2^n-1}\alpha_i \bra{i}\hat{Z}_k \ket{i}\right)^2\\
    &= \frac{1}{2^n}\left(\sum_{i=0}^{2^n-1}\alpha_i \bra{i}\left(\sum_{l=0}^{2^n-1} (-1)^{k \odot l}\ketbras{l}\right)\ket{i}\right)^2 \\
    &= \frac{1}{2^n}\left(\sum_{i,l=0}^{2^n-1}(-1)^{k \odot l}\alpha_i \braket{i|l}\!\braket{l|i}\right)^2 \\
    &= \frac{1}{2^n}\left(\sum_{i=0}^{2^n-1}(-1)^{k \odot i}\alpha_i\right)^2 \label{eq:Zm_schmidt_coeff}.
\end{align}
In particular, for $\hat{Z}_0 = I^{\otimes n}$, using the formula for the robustness of entanglement for pure states in \Cref{eq:robustness_pure_states} shows that
\begin{align}
    \braket{\Phi^{\hat{Z}_0}|\Psi^{\alpha}|\Phi^{\hat{Z}_0}} &=  \frac{1}{2^n}\left(\sum_{i=0}^{2^n-1}\alpha_i\right)^2 \\
    &= \frac{R(\Psi^{\alpha}) + 1}{2^n}.\label{eq:Zm_robustness}
\end{align}

Thus, the resulting state of the teleportation of an $n$-qubit state with a pure $2n$-qubit NME resource state $\Psi^{\alpha}$ is given by
\begin{align}
    \mathcal{E}_{\text{tel}}^{\Psi^{\alpha}}(\varphi)
    &=\sum_{\sigma \in \{I, X, Y, Z\}^{\otimes n}}\braket{\Phi^{\sigma}|\Psi^{\alpha}|\Phi^{\sigma}}\sigma\varphi\sigma \\
    &=\sum_{\sigma \in \{I, Z\}^{\otimes n}}\braket{\Phi^{\sigma}|\Psi^{\alpha}|\Phi^{\sigma}}\sigma\varphi\sigma \\
    &=\sum_{k=0}^{2^n-1}\braket{\Phi^{\hat{Z}_k}|\Psi^{\alpha}|\Phi^{\hat{Z}_k}}\hat{Z}_k\varphi\hat{Z}_k
\end{align}
where we first applied \Cref{eq:XY_zero} to the description of the teleportation result in \Cref{eq:tele2} and then used the fact that for each $\sigma \in \{I, Z\}^{\otimes n}$ exists a unique $k\in\mathbb{GF}(2^n)$ such that $\sigma = \hat{Z}_k$.

The derived teleportation outcome for an arbitrary state $\varphi$, utilizing the resource state $\Psi^{\alpha}$, allows us to rewrite the sum of the $2^n$ teleportations of the cutting protocol in this theorem (\Cref{eq:thrm2}).
Each teleportation within this sum employs a different unitary transformation $U_j$ as defined in \Cref{eq:mub_basis_transformation}.
The reformulation is presented as follows:
\begin{align}
    \sum_{j=0}^{2^n-1} U_j \mathcal{E}_{\text{tel}}^{\Psi^{\alpha}}(U_j^{\dagger}\varphi U_j)U_j^{\dagger} &= \sum_{j=0}^{2^n-1}U_j\left(\sum_{k=0}^{2^n-1}\braket{\Phi^{\hat{Z}_k}|\Psi^{\alpha}|\Phi^{\hat{Z}_k}}\hat{Z}_k U_j^{\dagger}\varphi U_j\hat{Z}_k\right)U_j^{\dagger} \\
    &= \sum_{k=0}^{2^n-1}\braket{\Phi^{\hat{Z}_k}|\Psi^{\alpha}|\Phi^{\hat{Z}_k}} \sum_{j=0}^{2^n-1} U_j\hat{Z}_k U_j^{\dagger}\varphi U_j\hat{Z}_k U_j^{\dagger}\\
    \begin{split}\label{eq:before_lemma_1}
        &=2^n\braket{\Phi^{\hat{Z}_0}|\Psi^{\alpha}|\Phi^{\hat{Z}_0}}\varphi \\
        &\phantom{=}+ \sum_{k=1}^{2^n-1}\braket{\Phi^{\hat{Z}_k}|\Psi^{\alpha}|\Phi^{\hat{Z}_k}} \sum_{j=0}^{2^n-1} U_j\hat{Z}_k U_j^{\dagger}\varphi U_j\hat{Z}_k U_j^{\dagger}
    \end{split}\\
    &=2^n\braket{\Phi^{\hat{Z}_0}|\Psi^{\alpha}|\Phi^{\hat{Z}_0}}\varphi + \sum_{k=1}^{2^n-1}\braket{\Phi^{\hat{Z}_k}|\Psi^{\alpha}|\Phi^{\hat{Z}_k}} \sum_{j=0}^{2^n-1} S_{j,k} \varphi S_{j,k} \label{eq:before_lemma_4}.
\end{align}
To obtain \Cref{eq:before_lemma_1}, we utilize the equality $U_j\hat{Z}_0 U_j^{\dagger}= U_j U_j^{\dagger}= I^{\otimes n}$.
Subsequently, we apply \Cref{lemma_1} in  \Cref{eq:before_lemma_4}.
Further, by utilizing \Cref{lemma_sum_S_im}, we can express \Cref{eq:before_lemma_4} as
\begin{align}
    \begin{split}
        \sum_{j=0}^{2^n-1} U_j \mathcal{E}_{\text{tel}}^{\Psi^{\alpha}}(U_j^{\dagger}\varphi U_j)U_j^{\dagger}&= 2^n\braket{\Phi^{\hat{Z}_0}|\Psi^{\alpha}|\Phi^{\hat{Z}_0}}\varphi \\
        &\phantom{=}+ \sum_{k=1}^{2^n-1}\braket{\Phi^{\hat{Z}_k}|\Psi^{\alpha}|\Phi^{\hat{Z}_k}} 2^n\sum_{l=0}^{2^n-1}\tr[\ketbras{l}\varphi] \ketbras{l\oplus k}
    \end{split}\\
    \begin{split}\label{eq:before_normalization}
        &= 2^n\braket{\Phi^{\hat{Z}_0}|\Psi^{\alpha}|\Phi^{\hat{Z}_0}}\varphi \\
        &\phantom{=}+ 2^n\sum_{l=0}^{2^n-1}\tr[\ketbras{l}\varphi]\sum_{k=1}^{2^n-1}\braket{\Phi^{\hat{Z}_k}|\Psi^{\alpha}|\Phi^{\hat{Z}_k}}  \ketbras{l\oplus k}.
    \end{split}
\end{align}

In the following, we normalize the coefficients $\braket{\Phi^{\hat{Z}_k}|\Psi^{\alpha}|\Phi^{\hat{Z}_k}}$ in the second line of \Cref{eq:before_normalization} such that their sum equals 1, enabling them to be interpreted as probabilities.
The states $\{\ket{\Phi^\sigma}\}_{\sigma \in \{I, X, Y,Z\}^{\otimes n}}$ form an orthonormal basis.
Thus, summing over the measurement outcomes with respect to this basis leads to
\begin{align}
    1 &= \sum_{\sigma \in \{I, X, Y, Z\}^{\otimes n}}\braket{\Phi^{\sigma}|\Psi^{\alpha}|\Phi^{\sigma}} \\
    &= \sum_{k=0}^{2^n-1}\braket{\Phi^{\hat{Z}_k}|\Psi^{\alpha}|\Phi^{\hat{Z}_k}}.
\end{align}
Consequently, by excluding $k=0$ from the summation, as is required for the normalization in the second line of \Cref{eq:before_normalization}, we have $\sum_{k=1}^{2^n-1}\braket{\Phi^{\hat{Z}_k}|\Psi^{\alpha}|\Phi^{\hat{Z}_k}}= 1-\braket{\Phi^{\hat{Z}_0}|\Psi^{\alpha}|\Phi^{\hat{Z}_0}}$.
We use this term to normalize the coefficients $\braket{\Phi^{\hat{Z}_k}|\Psi^{\alpha}|\Phi^{\hat{Z}_k}}$ resulting in probabilities 
\begin{align}\label{eq:prob}
    \Pr(k|\alpha)=\frac{\braket{\Phi^{\hat{Z}_k}|\Psi^{\alpha}|\Phi^{\hat{Z}_k}}}{1-\braket{\Phi^{\hat{Z}_0}|\Psi^{\alpha}|\Phi^{\hat{Z}_0}}}.
\end{align}
It is important to note that the theorem requires $\Psi^{\alpha}$ to be an NME state, i.e. $R(\Psi^{\alpha}) < 2^n - 1$.
This implies that $\braket{\Phi^{\hat{Z}_0}|\Psi^{\alpha}|\Phi^{\hat{Z}_0}} < 1$, ensuring the denominator remains non-zero. 
By inserting the probabilities of \Cref{eq:prob} in \Cref{eq:before_normalization} and accounting for the normalization, we obtain the following expression
\begin{align}
    &\sum_{j=0}^{2^n-1} U_j \mathcal{E}_{\text{tel}}^{\Psi^{\alpha}}(U_j^{\dagger}\varphi U_j)U_j^{\dagger}\\
    \begin{split}
        &=2^n\braket{\Phi^{\hat{Z}_0}|\Psi^{\alpha}|\Phi^{\hat{Z}_0}}\varphi \\
        &\phantom{=}+ 2^n(1-\braket{\Phi^{\hat{Z}_0}|\Psi^{\alpha}|\Phi^{\hat{Z}_0}})\sum_{l=0}^{2^n-1}\tr[\ketbras{l}\varphi]\sum_{k=1}^{2^n-1}  \Pr(k|\alpha)\ketbras{l\oplus k}.
    \end{split}
\end{align}
We now use  $\braket{\Phi^{\hat{Z}_0}|\Psi^{\alpha}|\Phi^{\hat{Z}_0}} = 2^{-n}(R(\Psi^{\alpha}) +1)$ as given in \Cref{eq:Zm_robustness} to express the coefficients of this sum in terms of the entanglement robustness:
\begin{align}
    &\sum_{j=0}^{2^n-1} U_j \mathcal{E}_{\text{tel}}^{\Psi^{\alpha}}(U_j^{\dagger}\varphi U_j)U_j^{\dagger}\\
    \begin{split}
        &=(R(\Psi^{\alpha}) + 1)\varphi \\
        &\phantom{=}+ (2^n-(R(\Psi^{\alpha}) + 1))\sum_{l=0}^{2^n-1}\tr[\ketbras{l}\rho]\sum_{k=1}^{2^n-1}\Pr(k|\alpha) \ketbras{l\oplus k}.
    \end{split}
\end{align}
By rearranging the equation, we obtain
\begin{align}
    \begin{split}
    \varphi &= \frac{1}{R(\Psi^{\alpha}) + 1} \sum_{j=0}^{2^n-1} U_j \mathcal{E}_{\text{tel}}^{\Psi^{\alpha}}(U_j^{\dagger}\varphi U_j)U_j^{\dagger} \\
    &\phantom{=}- \left(\frac{2^n}{R(\Psi^{\alpha}) + 1}-1\right)\sum_{l=0}^{2^n-1}\tr[\ketbras{l}\varphi]\sum_{k=1}^{2^n}\Pr(k|\alpha)  \ketbras{l\oplus k}.
    \end{split}
\end{align}
This provides the QPD in \Cref{theorem_decomposition}.
This QPD is comprised of the $2^n$ teleportation results (first sum), each weighted by $(R(\Psi^{\alpha}) + 1)^{-1}$.
Additionally, it incorporates a term, weighted by $2^n(R(\Psi^{\alpha})+1)^{-1} -1$, that corrects errors occurring in the teleportations (second sum).
The required operation to correct these errors is a single measure-and-prepare operation.
This operation entails an initial measurement in the computational basis, yielding result $l$, and then prepares with probability $\Pr(k|\alpha)$ the state $\ket{l \oplus k}$.
This probability can be rewritten to match the form as stated in \Cref{eq:thrm2} by applying $\braket{\Phi^{\hat{Z}_0}|\Psi^{\alpha}|\Phi^{\hat{Z}_0}} = 2^{-n}(R(\Psi^{\alpha}) +1)$ as given in \Cref{eq:Zm_robustness} and $\braket{\Phi^{\hat{Z}_k}|\Psi^{\alpha}|\Phi^{\hat{Z}_k}} = 2^{-n}\left(\sum_{i=0}^{2^n-1}(-1)^{k \odot i}\alpha_i\right)^2$ as given in \Cref{eq:Zm_schmidt_coeff}.
As a result, we obtain
\begin{align}
    \Pr(k|\alpha) &= \frac{\braket{\Phi^{\hat{Z}_k}|\Psi^{\alpha}|\Phi^{\hat{Z}_k}}}{1-\braket{\Phi^{\hat{Z}_0}|\Psi^{\alpha}|\Phi^{\hat{Z}_0}}}\\
    &= \frac{\braket{\Phi^{\hat{Z}_k}|\Psi^{\alpha}|\Phi^{\hat{Z}_k}}}{1- 2^{-n}(R(\Psi^{\alpha}) +1)}\\
    &= \frac{2^n\braket{\Phi^{\hat{Z}_k}|\Psi^{\alpha}|\Phi^{\hat{Z}_k}}}{2^n-1 -R(\Psi^{\alpha})}\label{eq:prob_proof}\\
    &=\frac{\left(\sum_{i=0}^{2^n-1}(-1)^{k \odot i}\alpha_i\right)^2}{2^n-1 -R(\Psi^{\alpha})}.
\end{align}

To show that this QPD achieves the optimal sampling overhead $\gamma^{\Psi^\alpha}(\mathcal{I}^{\otimes n})$ given in \Cref{theorem_overhead}, we compute its sampling overhead $\kappa$ as given in \Cref{eq:expectation} by summing over the absolute values of the circuits' coefficients.
The coefficients for each teleportation are given by $(R(\Psi^{\alpha}) + 1)^{-1}$ and the coefficient for the corrective measure-and-prepare operation is $2^n(R(\Psi^{\alpha})+1)^{-1} -1$.
Consequently, we obtain
\begin{align}
    &\sum_{j=0}^{2^n-1}\frac{1}{R(\Psi^{\alpha}) + 1} +  \frac{2^n}{R(\Psi^{\alpha}) + 1}-1 \\
    &= \frac{2^n}{R(\Psi^{\alpha}) + 1} +  \frac{2^n}{R(\Psi^{\alpha}) + 1}-1\\     
    &= \frac{2^{n+1}}{R(\Psi^{\alpha}) + 1} -1 \\
    &=  \gamma^{\Psi^\alpha}(\mathcal{I}^{\otimes n}) .
\end{align}
\end{proof}

\end{document}